\documentclass[lineno]{custom-arxiv}

\usepackage{amsmath}
\usepackage{caption} 
\usepackage{newtxtext}
\usepackage[subscriptcorrection]{newtxmath}
\usepackage[plain,noend]{algorithm2e}

\makeatletter

\usepackage{algorithmic}

\usepackage{ifthen}
\usepackage{pstricks}
\usepackage{subfigure}
\usepackage{enumerate}
\usepackage{ulem} 
\usepackage{verbatim}
\usepackage{multirow}
\usepackage{tabularx}
\usepackage{url}

\usepackage{xcolor,colortbl}

\let\emptyset\varnothing

\newcommand{\Te}{\mathcal{T}}
\newcommand{\Le}{\mathcal{L}}
\newcommand{\Se}{\mathcal{S}}
\newcommand{\Me}{\mathcal{M}}
\newcommand{\Ne}{\mathcal{N}}
\newcommand{\Pe}{\mathcal{P}}
\newcommand{\Ke}{\mathcal{K}}
\newcommand{\Ze}{\mathcal{Z}}

\begin{document}



\markboth{Valdez-Cabrera \& Willis}{Distances between Extension Spaces}

\title{\vspace{-1cm}Geometry of the space of phylogenetic trees with non-identical leaves}


\author{MAR\'IA A. VALDEZ-CABRERA}
\affil{Department of Biostatistics, University of Washington,\\ Seattle, WA, U.S.A.
\email{mariavc@uw.edu}}

\author{\and AMY D. WILLIS}
\affil{Department of Biostatistics, University of Washington,\\ Seattle, WA, U.S.A.
\email{adwillis@uw.edu}}

\maketitle






\newboolean{DEBUG}
\setboolean{DEBUG}{false}

\newrgbcolor{amycolor}{.5 .1 .99}
\ifthenelse {\boolean{DEBUG}}
{\newcommand{\amy}[1]{{\amycolor{[\@Amy: #1]}}}}
{\newcommand{\amy}[1]{}}

\newrgbcolor{mariacolor}{.392 .585 .93}
\ifthenelse {\boolean{DEBUG}}
{\newcommand{\maria}[1]{{\mariacolor{[\@María: #1]}}}}
{\newcommand{\maria}[1]{}}

\ifthenelse {\boolean{DEBUG}}
{\newcommand{\mariadraft}[1]{{\mariacolor{#1}}}}
{\newcommand{\mariadraft}[1]{#1}}

\begin{abstract}
Phylogenetic trees summarize evolutionary relationships. The Billera-Holmes-Vogtmann (BHV) space for comparing phylogenetic trees has many elegant mathematical properties, but it does not encompass trees with differing leaf sets. To overcome this, we introduce Towering space: a complete metric space that extends BHV space to trees with non-identical leaf sets. Towering space is a structured collection of BHV spaces connected via pruning and regrafting operations. We study the geometry of paths in Towering space and present an algorithm for computing metric distances. By addressing a major limitation of BHV space, Towering space facilitates the analysis of modern phylogenetic datasets such as multi-domain gene trees.
\end{abstract}



\section{Introduction}
\label{sec:Introduction}

Phylogenetic trees represent evolutionary relationships between organisms. While a single phylogenetic tree can describe the shared evolution of large, complex eukaryotes, the genes of simple organisms may meaningfully differ in their evolutionary histories. This is in part because 
genes are more commonly transferred, duplicated, gained, or lost in bacteria, archaea, and viruses than in eukaryotes \citep{Jain.etal1999, Koonin.etal2001}. In these organisms, different genes have distinct evolutionary histories, and a collection of ``gene trees'' offers a more complete summary of phylogenetic relationships than a single tree.  

Unfortunately, analyzing collections of phylogenetic trees remains challenging in practice. 
A major advancement in mathematical phylogenetics was the introduction of the BHV tree space \citep{BILLERA2001}, a metric space for phylogenetic trees that accounts for differences between trees with respect to both their topologies and their branch lengths. Geodesics in BHV space are smooth, with trees along the path deforming in a locally continuous manner. The polynomial-time algorithm of \cite{OwenMegan2011} cemented the impact of BHV space, facilitating the development of BHV-based tools for estimating means, variances, principal paths and densities  \citep{blow, bbb, willis2018uncertainty, brown2020mean, nye11, nyepca, weyenberg2014kdetrees}, constructing confidence sets \citep{willis2019confidence}, and clustering \citep{gori2016clustering}. 
Unfortunately, BHV space suffers a major limitation that precludes its applicability to modern datasets: the BHV distance is only defined between trees with the same leaf set. If an organism is missing even one gene (such as due to gene loss or data incompleteness), either the gene tree or the organism must be discarded from a BHV-based analysis. 
To address this shortcoming, we introduce a metric space for phylogenetic trees with non-identical leaf sets. Our proposal is closely connected to BHV space, and shares its property that distances account for both topological and edge length differences between trees.

Previous attempts to extend BHV space to accommodate trees with non-identical leaf sets considered combinatorial methods to transition between BHV tree spaces \citep{ren2017combinatorial}, and constructed representative sets within the largest possible BHV space for all trees \citep{GrindstaffGillian2019RoPL}. However, these approaches yield dissimilarity measures, not metrics.  
Our approach yields a metric, and paths have a biological interpretation reflecting gene gain and loss, in addition to sequence divergence. 
Because our proposed paths traverse nested BHV spaces of differing dimension, we call it the \textbf{Towering space}.  

While most phylogenetic tree distances are defined only for trees with the same leaf set, our proposal is by no means the first distance that accommodates non-identical leaf sets. The Robinson-Foulds distance --- a well-known and simple distance that reflects only topological differences --- has multiple extensions to accommodate non-identical leaf sets (reviewed in \cite{LiWanlin2024Copt}). However, to our knowledge, no (proper) metric has been defined for trees with non-identical leaf sets that accounts for differences in both topology and branch lengths. Our work fills this gap. 

Our paper is structured as follows: Section \ref{sec:Background} introduces notation and summarizes key elements of BHV space. Section \ref{sec:PreliminaryStructures} discusses the geometry of connections between BHV spaces via leaf prunings and regraftings. Section \ref{sec:BuildingUp} formally defines Towering space and describes the geometry of geodesics, which may traverse multiple BHV spaces. We conclude the paper in Section \ref{sec:discussion} with a discussion of the implications of the results, applications to broader biological studies, and future work. Proofs of all results can be found in the Supplementary Material.

\section{Background}
\label{sec:Background}

\subsection{Preliminaries}

Phylogenetic trees are labeled tree-graphs that describe the shared evolutionary history of a set of organisms. The branching of the tree reflects the structure of the organisms' evolutionary relationships, and branch lengths represent the distance between divergence events. Our proposal is equivalent for trees with and without a ``root'' organism. For generality, we proceed by considering unrooted trees. 

\begin{definition}
    \label{Def:PhyloTree}
    A \textbf{phylogenetic tree} $T$ on a set of organisms $\mathcal{L}$ is a connected edge-weighted acyclic graph whose \textbf{leaves} (nodes of degree one) are labeled by the elements of $\mathcal{L}$, and all interior nodes have degree greater or equal to 3. All edge-weights are non-negative values referred to as \textbf{lengths}. The branching pattern of the graph gives the \textbf{topology} of the tree. 
\end{definition}

Removing an edge from a tree splits its leaves into two groups, and such a split of leaves uniquely identifies the edge in the tree.  
Throughout we use the notation $\mathcal{G} \big| \mathcal{L} \setminus \mathcal{G}$, with $\mathcal{G} \subset \mathcal{L}$, to represent the \textbf{split} on $\mathcal{L}$ that separates the leaves in $\mathcal{G}$ from $\mathcal{L} \setminus \mathcal{G}$. Edges that induce the same split on topologically distinct trees are considered to be the same edge. Therefore, we use the terms edge and split interchangeably. Edges connecting to leaves are called the \textbf{external} edges, which correspond to splits $|\mathcal{G}| = 1$ or $|\mathcal{L}|-1$, while all others, with $2 \leq |\mathcal{G}| \leq |\mathcal{L}| - 2$, are the \textbf{internal} edges. The topology of a tree is fully defined by all its internal splits \citep[Theorem 2.34]{AlgStatCompBio}. We use $\mathcal{S}(T)$ to refer to the set of internal splits in $T$ (equivalently, the topology of $T$), 
while $\mathcal{H}(T)$ refers to the set of external splits and $\mathcal{P}(T) = \mathcal{S}(T) \cup \mathcal{H}(T)$ the entire set of edges. 
Only certain combinations of splits can be present on a tree. We say two splits are \textbf{compatible} \citep[Section 3.2]{BILLERA2001} if they can simultaneously be present on a tree. 


\subsection{The Billera-Holmes-Vogtmann tree space}

\citet{BILLERA2001} introduced the BHV tree space $(\mathcal{T}^{\mathcal{N}}, d_{\text{BHV}})$, where $\mathcal{T}^{\mathcal{N}}$ denotes the set of all trees with leaf set $\mathcal{N}$, and $d_{\text{BHV}}: \mathcal{T}^{\mathcal{N}} \times \mathcal{T}^{\mathcal{N}} \mapsto \mathbb{R}_{\geq 0}$ is a distance function. This metric space is Hadamard \citep{bacak2014computing}, 
which guarantees continuous paths connecting any pair of trees, as well as a unique path achieving the minimum length, called the \textit{geodesic}. The space we define in Section \ref{sec:BuildingUp} is comprised of connected BHV spaces. 

Consider all possible topologies for trees with leaf set $\mathcal{N}$ of cardinality $n = |\mathcal{N}|$. Equivalently, consider all possible subsets of internal splits on $\Ne$ that are pairwise compatible. For one such subset $S$, assign an order to its elements ($S = \{s_1, s_2, \hdots, s_m\}$) in any consistent way; for example, a lexicographic order can be assigned to the leaves and the edges may inherit this order. A tree with the topology given by $S$ may be represented by a $(n+m)-$dimensional non-negative vector, where the first $n$ coordinates represent the lengths of the external edges and the last $m$ coordinates represent the lengths of the internal branches in the given order (Figure \ref{fig:OneTopology}). The vectors representing possible trees with topology $S$ form an orthant in $\mathbb{R}^{n+m}$, called the \textbf{topology orthant} and denoted by $\mathcal{O}(S)$  \citep[Section 2.1]{OwenMegan2011}. 

\begin{figure}
    \centering
    \subfigure[]{\includegraphics[width = 0.4\textwidth]{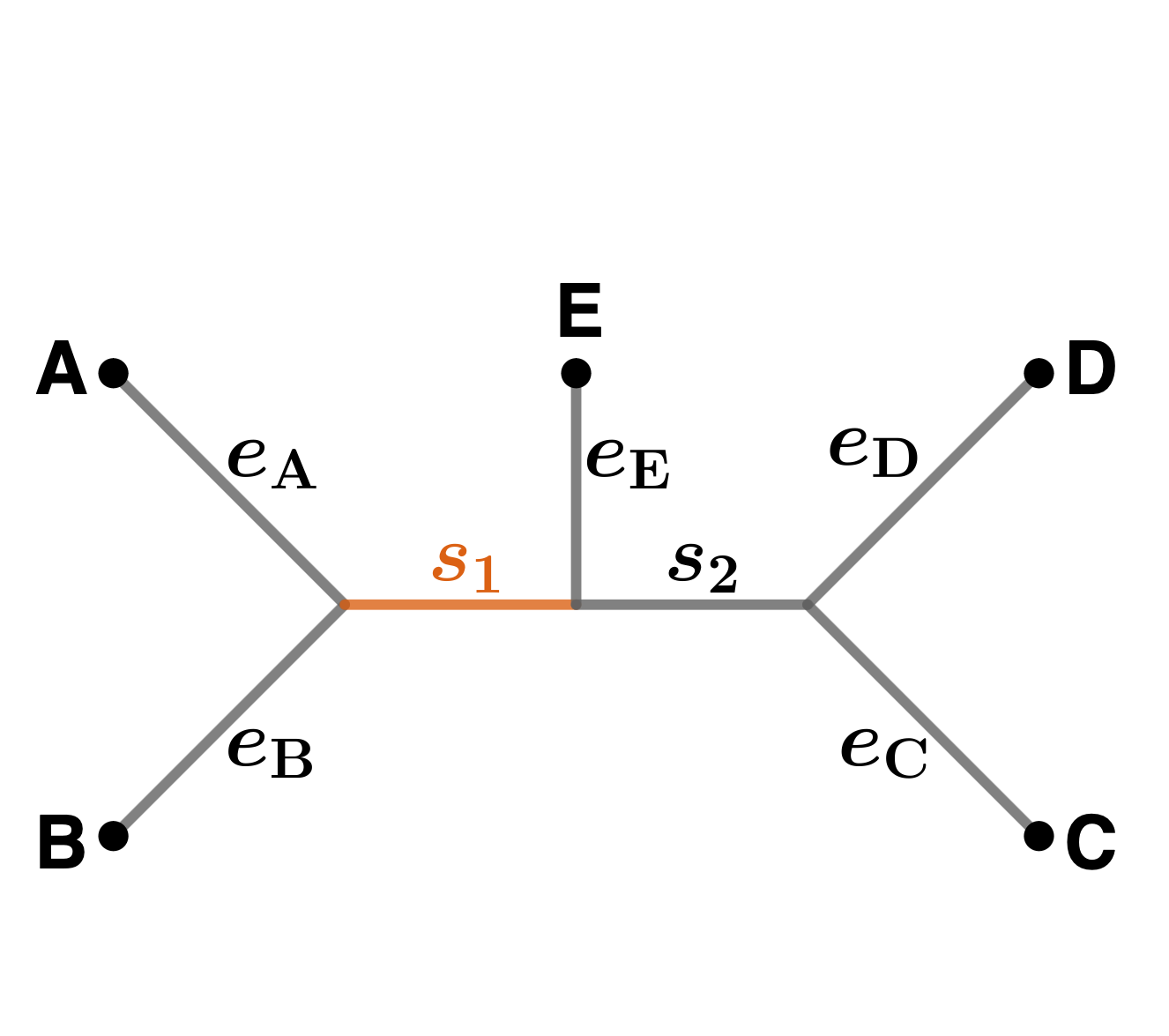}
    \label{fig:OneTopologyA}}
     \subfigure[]{\includegraphics[width = 0.46\textwidth]{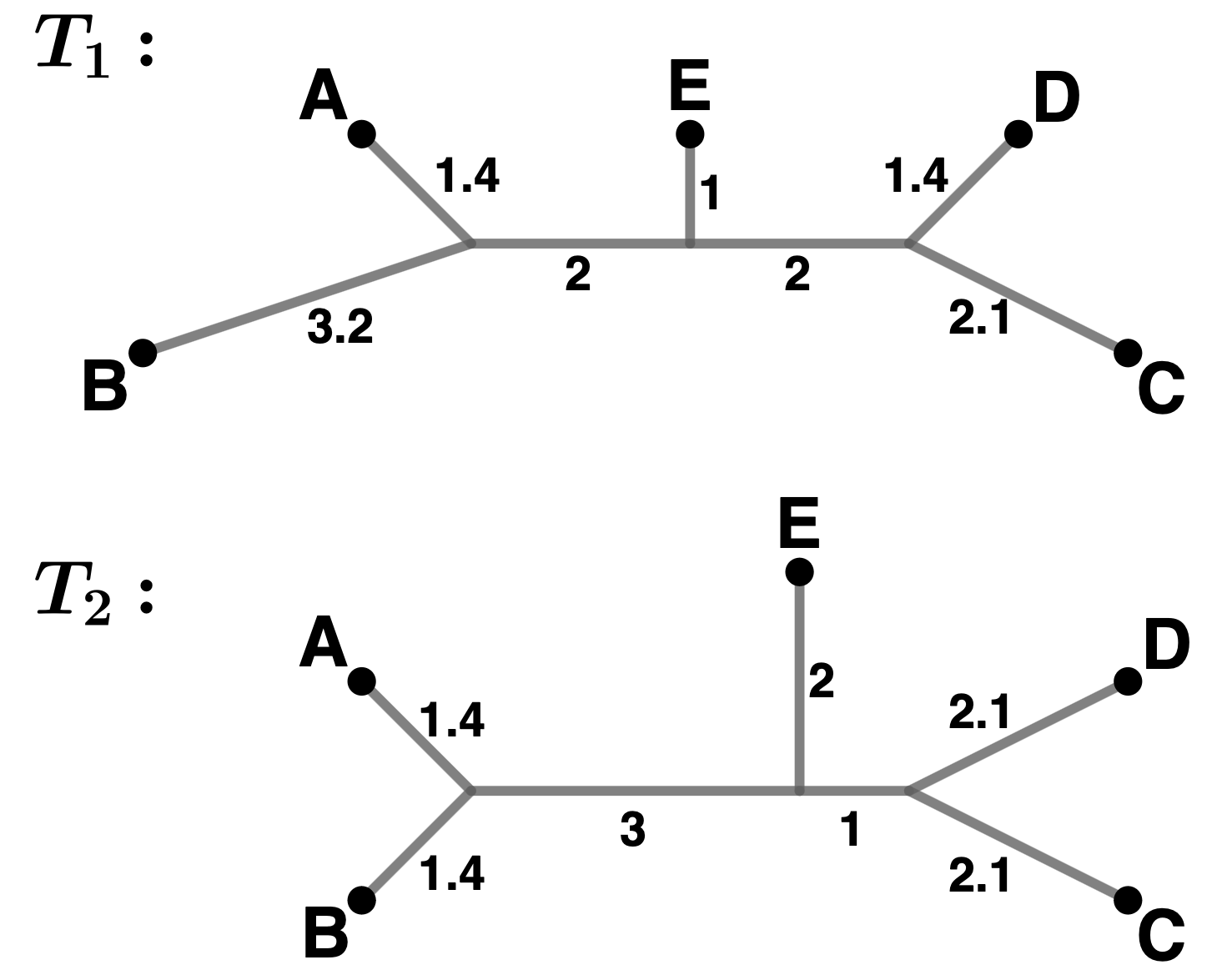}
     \label{fig:OneTopologyB}}
     	
    \caption{(a) A phylogenetic tree with leaf set $\mathcal{N} = \{A, B, C, D, E\}$ and internal split $s_1 = \{A,B\} \big| \{C,D,E\}$ highlighted in orange. This topology can be summarized by 7-dimensional vectors $(e_A, e_B, e_C, e_D, e_E, s_1, s_2)$. (b) Two different trees with the same topology. In its topology orthant, $T_1$ corresponds to the vector $(1.4, 3.2, 2.1, 1.4, 1, 2, 2)$ while $T_2$ corresponds to the vector $(1.4, 1.4, 2.1, 2.1, 2, 3, 1)$.}
    \label{fig:OneTopology}
\end{figure}

A tree $T$ with an internal edge $s \in \mathcal{S}(T)$ of length zero represents the same graph as the tree without $s$ among its internal edges, provided all other edges remain unchanged (Figure \ref{fig:OrthantFacesA}). These two trees are therefore considered to be the same. Furthermore, since two additional orthants contain this tree, these three orthants are connected at this point. More generally, given a set of internal splits $S$ with $m$ elements, any topology orthant $\mathcal{O}(S')$ corresponding to a proper subset $S' \subset S$ of these internal edges can be viewed as a \textbf{face} of the larger topology orthant $\mathcal{O}(S)$  \citep[Section 2]{BILLERA2001}, corresponding to the trees with length zero in all edges in the difference $S \setminus S'$. 
Two topology orthants $\mathcal{O}(S_1)$ and $\mathcal{O}(S_2)$ such that $S := S_1 \cap S_2$ are glued along the face corresponding to the topology orthant $\mathcal{O}(S)$ (Figure \ref{fig:OrthantFacesB}). The only topology orthants that are not faces of another topology orthant are those of maximum dimension $2n-3$, corresponding to the topologies of binary trees.

\begin{remark}
	\label{Remark:dropInternalEdgesZero}
	For any phylogenetic tree, there is a unique representation of its topology corresponding to the subset of internal edges $S$ with positive lengths in the tree; i.e. dropping any internal edge of length zero from the set of internal edges. Henceforth, when referring to a tree $T$, $S(T)$ contains strictly positive edges unless indicated otherwise. We use $\mathcal{O}(T)$ to refer to the lowest-dimensional orthant containing $T$.
\end{remark}

A \textit{continuous} path can be traced between any two trees, where a path is a piecewise connected curve, each piece fully contained within a single orthant, and contiguous pieces connected through common faces of these orthants. Since each piece is a curve in a $\mathbb{R}^{n+m}$ subspace, the lengths of the edges of the trees along these pieces change smoothly. Meanwhile, the transition from an orthant to the next requires some of these lengths to reduce to zero, at which point those zero-length edges are replaced by others. Thus, a path in BHV space represents how to gradually morph one tree into another by slowly modifying the edge lengths and topology. BHV space is a complete geodesic space with non-positive curvature  \citep[Lemma 4.1]{BILLERA2001}, and thus the shortest path between any two trees (the \textbf{geodesic}) is unique. The \textbf{BHV distance} between the two trees $d_{\text{BHV}}(T_1, T_2)$ is the length of the geodesic that connects them, as this geodesic describes the most efficient method to transform $T_1$ into $T_2$. 


\subsection{Geodesics in BHV space}
\label{sec:OwenProvan}

The geometrical properties of BHV space enabled the construction of an algorithm to compute distances that is polynomial-time in the number of leaves $n$ \citep{OwenMegan2011}. This algorithm results in a closed-form expression for the length of a geodesic. Here, we describe the elements of this expression, as they will be relevant when discussing distances in  Towering space. 

\begin{definition}
    \label{Def:PathSpace}
     \citep[Definition 3.3]{owen1} For trees $T_1$ and $T_2$, consider the sets of internal splits $S_1 = \mathcal{S}(T_1)$, $S_2 = \mathcal{S}(T_2)$ and the set of ``common'' internal splits $C \subset S_1 \cup S_2$ such that $s \in C$ implies $s$ is pairwise compatible with all edges in $S_1 \cup S_2$. For a sequence of subsets $S_1 = G_0 \supset G_1 \supset \hdots \supset G_{k-1} \supset G_{k} = C$ and $C = F_0 \subset F_1 \subset \hdots \subset F_{k-1} \subset F_{k} = S_2$ such that $G_i \cup F_i$ is a set of pairwise compatible splits, denote by $O_i = \mathcal{O}(G_i \cup F_i)$ for all $i \in \{0,\hdots, k\}$. The sequence of orthants,
\begin{equation*}
   \mathcal{O}(T_1) = O_0  \rightarrow O_1 \rightarrow \hdots \rightarrow O_k = \mathcal{O}(T_2),
\end{equation*}
is a \textbf{path space} from $T_1$ to $T_2$. 
\end{definition}


\begin{definition}
    \label{Def:Support}
     \citep[Section 2.3]{OwenMegan2011} Consider a path space $O_0 \rightarrow \hdots \rightarrow O_k$. When transitioning from $O_{i-1}$ to $O_{i}$, denote by $A_i = G_{i-1} \setminus G_{i}$ the set of dropped edges and by $B_i = F_{i} \setminus F_{i-1}$ the set of added edges. For ordered sets $\mathcal{A} = \left\{A_1, \hdots, A_k\right\}$ and $\mathcal{B} = \left\{B_1, \hdots, B_{k}\right\}$, the pair $(\mathcal{A}, \mathcal{B})$ gives the \textbf{support} of the path space. Each $(A_i,B_i)$ is also called a \textbf{support pair}.
\end{definition}

\begin{figure}
    \centering
    \subfigure[]{\includegraphics[width = 0.45\textwidth]{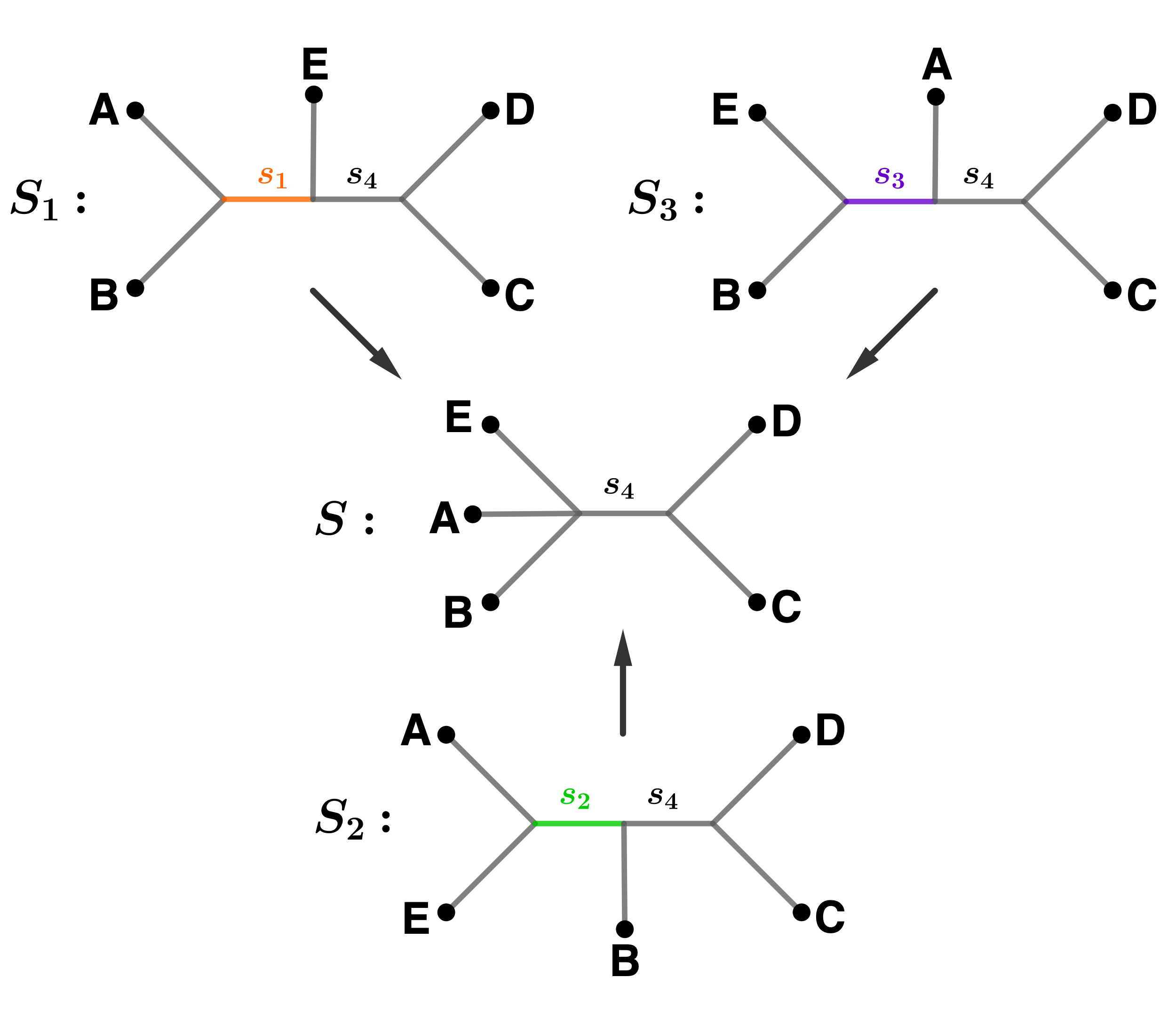}
    \label{fig:OrthantFacesA}}
   \subfigure[]{\includegraphics[width = 0.49\textwidth]{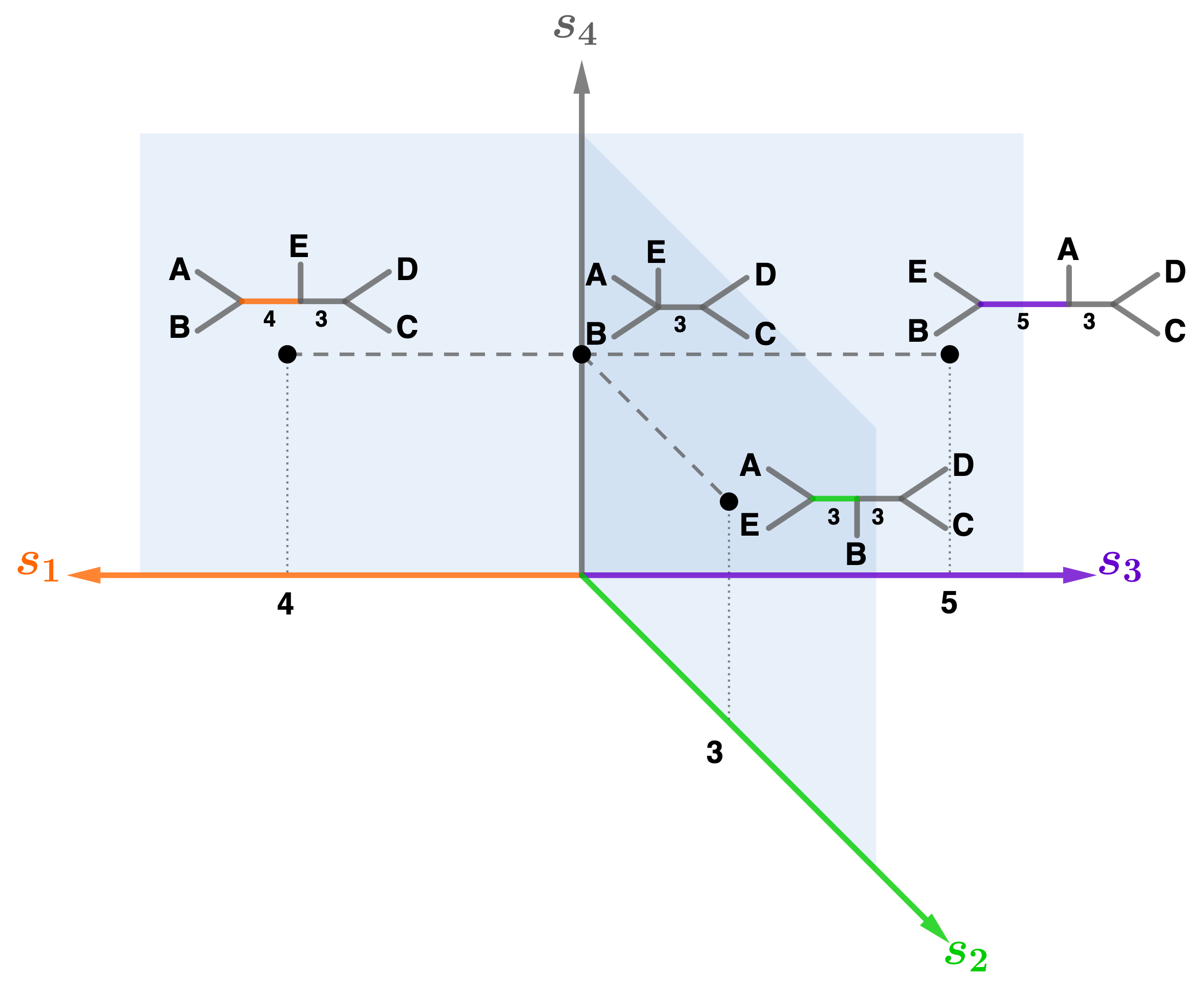}
   \label{fig:OrthantFacesB}}
    \caption{(a) Topologies $S_1$, $S_2$ and $S_3$ share topology $S$ as a face. The topology $S_i$ is composed of internal edges $S_i = \{s_i, s_4\}$, while the topology $S$ only has $S = \{s_4\}$; thus $\mathcal{O}(S)$ is a face of $\mathcal{O}(S_i)$ for $i=1, 2, 3$. (b) The topology orthants $\mathcal{O}(S_1)$, $\mathcal{O}(S_2)$ and $\mathcal{O}(S_3)$ are ``glued'' at $\mathcal{O}(S)$. The dimensions arising from external edges are not shown.}
    \label{fig:OrthantFaces}
\end{figure}

A path space describes an efficient way to move from the topology orthant of $T_1$ to the topology orthant of $T_2$ through connected orthants. By \citet[Proposition 4.1]{BILLERA2001}, we know the geodesic will be fully contained in a path space. 
Given the support for that path space, the length of the geodesic can be expressed by the $L^2$ norm of each $A_i$ and $B_i$ in the support and the difference in length of the common edges, including the external edges $H = \mathcal{H}(T_1) = \mathcal{H}(T_2)$. Denoting by $|p|_{T}$ the length of the edge $p$ in tree $T$ and by $||P||_{T} = \sqrt{\sum_{p \in P} |p|_T^{2}}$ the $L^2$ norm of the lengths in $T$ of all edges in the set of splits $P$, the length of the geodesic can be written in these terms by

\begin{equation}
\label{eq:FirstDistance}
        d_{\text{BHV}}(T_1, T_2) = \sqrt{\sum_{i =1}^{k}\left(||A_i||_{T_1} + ||B_i||_{T_2}\right)^2 + \sum_{s\in K} \left(|s|_{T_1} - |s|_{T_2}\right)^2} 
\end{equation}
where $K = C \cup H$  (see details in \citet[Section 4]{OwenMegan2011} and \citet[Theorem 1.2]{MillerOwenProvan}). 

We conclude by noting that once 
the best way to group the uncommon edges between the two trees through the support pairs is found, the distance is fully determined. The length of each edge in both trees contributes to this distance depending on whether they are a common or an uncommon edge. 
As we will see in Section \ref{sec:ShortPathsHighBHVspaces}, the mechanisms that determine the distance in our metric space will follow a similar pattern.

\section{Transitions between BHV spaces}
\label{sec:PreliminaryStructures}

Although BHV space offers a practical and intuitive metric for comparing trees, it is only defined when trees have exactly the same leaves. This is a significant limitation for many important applications, including the study of genes shared by distantly related organisms and the origins of complex, multicellular life. \citet{ren2017combinatorial} addressed the restriction of BHV space to trees with identical leaves by introducing structured subsets in a BHV space formed by adding leaves to trees with fewer leaves and studied their combinatorial properties. 
Building on this, \citet{GrindstaffGillian2019RoPL} introduced Extension spaces, subspaces of BHV space that represent trees with fewer leaves, and proposed a dissimilarity measure between trees based on these spaces. More recently, \citet{VW_Extensions} developed a linear-optimization-based algorithm to compute shortest paths between Extension spaces, extending the dissimilarity concept to arbitrary tree pairs. However, this dissimilarity is not a distance, precluding the development of sophisticated analysis tools that rely on the triangle inequality, including (Fr\'echet) averages and principal paths. 

In this section, we propose a new approach to connecting BHV spaces. Our approach overcomes limitations of prior work, including that it leads to a true distance (rather than a dissimilarity measure) with favorable geometric properties. This new metric also supports straightforward, BHV-like computations, avoiding  gradient-based optimization procedures like those in \citet{VW_Extensions}.

\subsection{Maps between topologies}
 
Removing and adding leaves from a tree is a common operation in phylogenetics and combinatorics,  and forms the basis for the subtree prune-and-regraft (SPR) distance between tree topologies \citep{HeinJotun1996Otco}. 
As in \citet{ren2017combinatorial, GrindstaffGillian2019RoPL}, we consider the 
\textit{tree dimensionality reduction map} (TDR) from \citet[Definition 4.1]{zairis2016genomic}, which removes leaves from a tree graph (Figure \ref{fig:Prunable_Projection}, right). The removal results in edges connected by two-degree nodes that must be merged to give a valid tree. 
In the TDR map, the lengths of edges resulting from merging are defined to be the sum of the adjacent edges, or equivalently, the $L^1$ norm of the lengths of these edges. 
Using a norm to define the new lengths results in favorable properties, including consistent lengths under sequential leaf removals (a result of commutativity and associativity of norms). 
 
We employ the modified version of the TDR map given in \citet[Definition 2.4]{GrindstaffGillian2019RoPL}, which presents the TDR map as a function on edges rather than trees. As we will see later (Definition \ref{Def:LeafPrune}), we then employ the $L^2$ norm to assign lengths to merged edges, instead of the $L^1$ norm. The advantageous geometrical properties of $L^2$ merging will become apparent in Section \ref{sec:BuildingUp}. 

\begin{definition}
	\label{Def:TDRmap_splits}
     For $\Le' \subseteq \Le$ and a split $s = \mathcal{G} \big| \Le\setminus\mathcal{G}$ on $\mathcal{L}$, the \textbf{tree dimensionality reduction} (TDR) map is the projection function $\Psi_{\Le'}$ that returns the split on $\Le'$ that separates all leaves of $\mathcal{L}'$ in $\mathcal{G}$ from those not in $\mathcal{G}$; that is, $\Psi_{\Le'}(s) = \Le'\cap\mathcal{G} \big| \Le' \setminus \mathcal{G}$. 
\end{definition}

\begin{remark}
    If either $\mathcal{G}\cap\mathcal{L}' = \emptyset$ or $\mathcal{L}' \setminus \mathcal{G} = \emptyset$, the projection $\Psi_{\mathcal{L}'}(s)$ is not a valid edge. This happens when this split is not part of any path inside the tree connecting leaves in $\mathcal{L}'$. In these cases we write $\Psi_{\mathcal{L}'}(s) = \emptyset$.
\end{remark}

Our proposed leaf pruning operation differs slightly from the literature in that we restrict the set of trees on which this operation can be performed. Similar to dropping zero-length \textit{internal} edges when transitioning between \textit{topology orthants} in BHV space, only zero-length \textit{external} edges can be dropped to transition between BHV spaces. The trees where leaf pruning may be performed are defined below.   

\begin{definition} 
	\label{Def:prunable_All}
	Given a subset of leaves $\Me \subset \Le$ and $T \in \Te^{\Le}$, the leaves in $\Me$ are \textbf{prunable} from $T$ if each external edge to a leaf in $\Me$ is of length zero, and every internal edge of $T$ maps to an internal split of the leaves $\Le' = \Le \setminus \Me$ under the TDR map. That is, for every $s \in \Se(T)$, $\Psi_{\Le'}(s) = \Le_1 \big| \Le' \setminus \Le_1$ such that $|\Le_1|, |\Le \setminus \Le_1| \geq 2$.

If $\Me$ is not  prunable from $T$, there is a set of edges (external edges of positive length and internal edges) preventing $\Me$ from being prunable from $T$. We denote this set as $P^{\downarrow \Me}(T)$:
	\begin{equation*}
	\label{Eq:Trimmable_ProjectedEdges}
	P^{\downarrow \Me}(T) = \left\{p \in \Pe(T) \mid \Psi_{\Le'}(p) = \emptyset \text{ or } \Psi_{\Le'}(p) = \{\ell\} \big| \Le' \setminus \{\ell\} \text{ for some leaf } \ell \right\}.
	\end{equation*}
Of note, all edges that belong to $P^{\downarrow \Me}(T)$ are of the form $\Me' \big| \Le \setminus \Me'$ and $\Me' \cup \{\ell\} \big| \Le \setminus \left(\Me' \cup \{\ell\}\right)$ for some non-empty subset $\Me'\subseteq \Me$. 

\end{definition}
	
\begin{figure}
    \centering
    	\includegraphics[width = 0.9\textwidth]{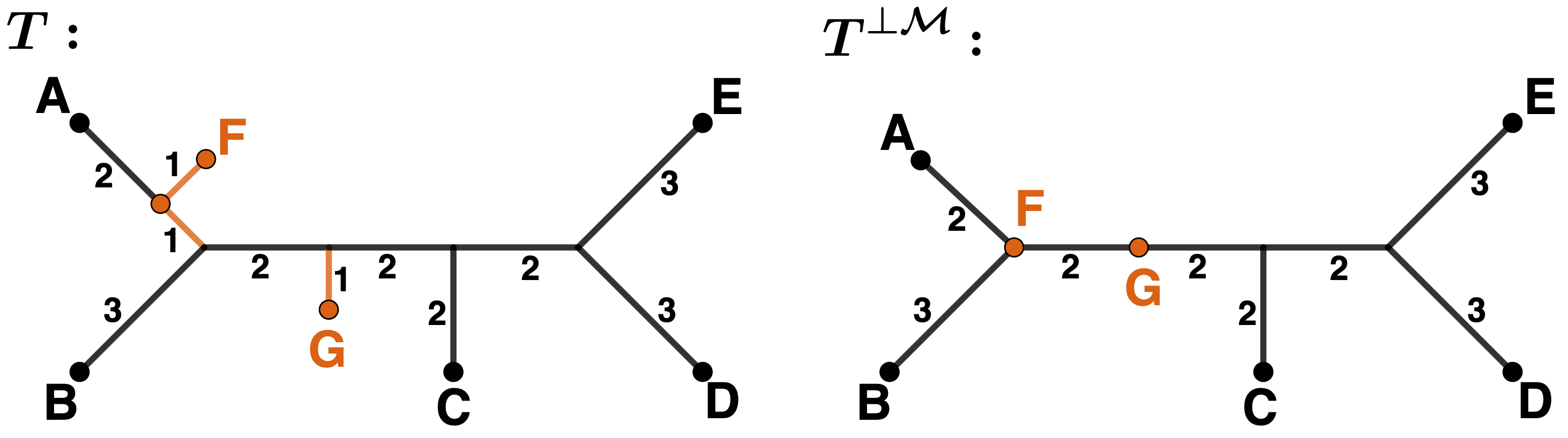}
    \caption{The projection of $T$ (left) with leaf set $\Le = \{A,B,C,D,E,F,G\}$ onto the $\Me$-trimmable space for $\Me = \{F,G\}$ (right). The edges in $P^{\downarrow \Me}(T)$ preventing $\Me$ from being prunable from $T$ are shown in orange. However, the same set is prunable from $T^{\perp \Me}$. The distance between the two trees is $\sqrt{3}$. When $\{F,G\}$ is pruned from $T^{\perp \Me}$, the edge length of $\{A, B\} | \{C, D, E\} $ is the $L^2$ norm of the adjacent edges, $\sqrt{2^2 + 2^2}$.}
    \label{fig:Prunable_Projection}
\end{figure}

This definition ensures that when leaves are removed, all internal edges map to internal edges. This restriction aids in maintaining a strictly positive distance between \textit{non-deformed trees} (trees with strictly positive edges) that are clearly distinct from each other. Later, we will use leaf prunings and regraftings as zero-distance operations. That is, two trees connected through these operations will be at distance zero from each other. If a leaf could be removed from or regrafted onto an external edge, then leaves could be just ``swapped.'' For example, without the additional restriction, all trees in Figure \ref{fig:SeveralJumps} would be at distance zero, which is clearly undesirable. 

\begin{figure}
    \centering
    	\includegraphics[width = 0.95\textwidth]{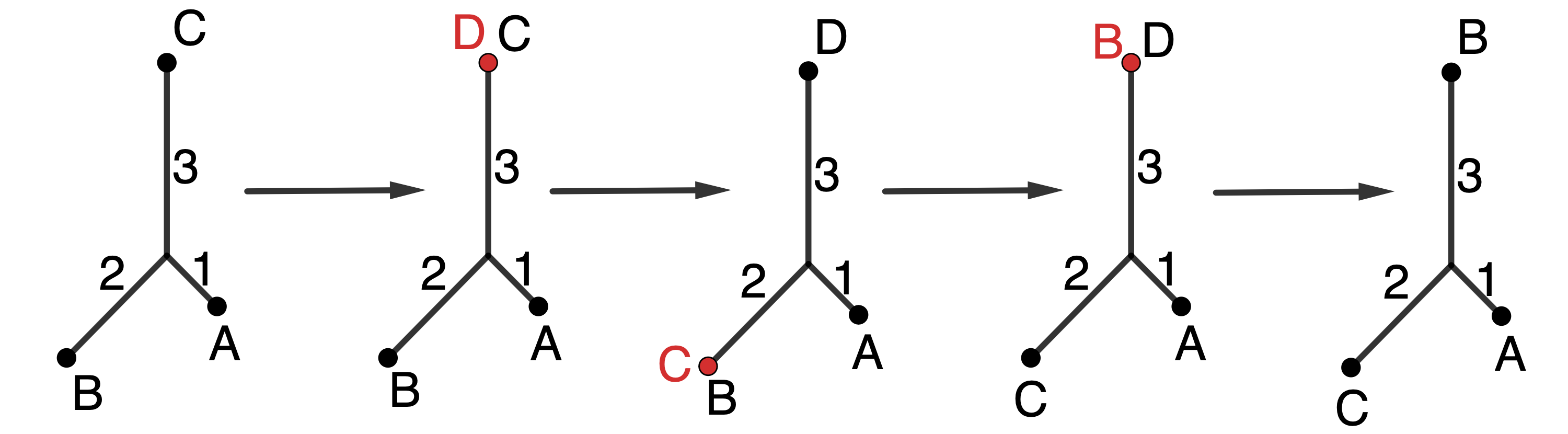}
    \caption{We impose the restriction that leaves can only be pruned from internal, not external, edges. Without this restriction, trees that are clearly distinct would be at distance zero, as in this example. A new leaf $D$ is first regrafted directly next to leaf $C$, then $C$ is pruned and regrafted next to $B$, which is later pruned and regrafted next to $D$, then $D$ is finally pruned. The resultant tree is clearly distinct from the initial tree, but all transitions were zero distance. 
	}
    \label{fig:SeveralJumps}
\end{figure}

Given two BHV spaces $\Te^{\Le}$ and $\Te^{\Le'}$, we say $\Te^{\Le}$ is a \textbf{higher level} than $\Te^{\Le'}$ (and $\Te^{\Le'}$ is a \textbf{lower level} than $\Te^{\Le}$) whenever $\Le'\subset \Le$. 
As trees with prunable leaves offer junctions between BHV levels, they comprise an important subspace, which we now consider, followed by some convenient properties.

 \begin{definition}
	\label{Def:Trimmable_subspace}
	Consider a subset of leaves $\Me \subset \Le$ and BHV space $\Te^{\Le}$. The \textbf{$\Me$-trimmable subspace}, denoted by $\Ze_{\Me}^{\Le}$, is the subset of trees in $\Te^{\Le}$ where $\Me$ is  prunable. 
\end{definition}

\begin{lemma}
	\label{lemma:Trimmable_convex}
	For any BHV space $\Te^{\Le}$ and a subset of leaves $\Me \subset \Le$, 
	$\Ze_{\Me}^{\Le}$ is a convex and closed subspace. 
\end{lemma}

Later on, efficient moves towards trees in $\Ze_{\Me}^{\Le}$ will be useful for finding short paths between trees in different BHV levels. The following result concerns the projection of a tree onto $\Ze_{\Me}^{\Le}$ (i.e. the point in the subspace that is closest to the tree), and follows directly from \citet[Proposition 2.6]{sturm}.

\begin{corollary}
	\label{cor:ProjectionGood}
	For every tree $T \in \Te^{\Le}$ there exists a unique tree  $T^{\perp \Me} \in \Ze_{\Me}^{\Le}$ such that $d_{\text{BHV}}(T,T^{\perp \Me}) = \inf_{t \in \Ze_{\Me}^{\Le}} d_{\text{BHV}}(T, t)$. Moreover,  $d_{\text{BHV}}^2(T,t) \geq d_{\text{BHV}}^2(T,T^{\perp \Me}) + d_{\text{BHV}}^2(T^{\perp \Me}, t)$ for every $t \in \Ze_{\Me}^{\Le}$. 
\end{corollary}

The projection of any tree $T \in \Te^{\Le}$ onto the $\Me$-trimmable subspace can be found directly by focusing on the edges in $P^{\downarrow \Me}(T)$ (see Figure \ref{fig:Prunable_Projection}). We now make this connection explicit and show how to compute the distance between trees and their projections. 

\begin{lemma}
	\label{lemma:ProjectionCharacterize}
	For $T \in \Te^{\Le}$ and a subset of leaves $\Me \subset \Le$, consider the set of edges $P^{\downarrow \Me}(T)$. The projection of $T$ onto $\Ze_{\Me}^{\Le}$ is the tree $T^{\perp \Me}$ with topology given by internal edges 
    $$\Se(T^{\perp \Me}) = \left\{s \in \Se(T) \mid s \notin P^{\downarrow \Me}(T) \right\},$$
    with the lengths of these edges coinciding with the corresponding lengths in $T$; i.e., $|s|_{T^{\perp \Me}} = |s|_{T}$ for every $s \in \Se(T^{\perp \Me})$. Similarly, all external edges not in $P^{\downarrow \Me}(T)$ are of the same length as the external edges in $T$, but all external edges $P^{\downarrow \Me}(T)$ are of length zero.  Furthermore, $d_{\text{BHV}}(T, T^{\perp \Me}) = ||P^{\downarrow \Me}(T)||_{T}.$
\end{lemma}

\subsection{Pruning and regrafting} \label{subsec:prune_regraft}

Having introduced a mapping between tree topologies in different BHV spaces and the subspaces where leaves may be removed, we can now formally define leaf pruning and regrafting. These operations connect trees in $\Ze_{\Le \setminus \Le'}^{\Le}$ and $\Te^{\Le'}$ such that BHV distances across different BHV levels remain consistent and have nice geometric properties. To this end, the last step in the pruning operation assigns new edge lengths by using the $L^2$ norm of the lengths of corresponding edges in the original tree. 

\begin{definition}
	\label{Def:LeafPrune}
	Consider $T\in \Te^{\Le}$ with a  prunable subset of leaves $\Me$, and denote $\Le' = \Le \setminus \Me$. The \textbf{pruning of $\Me$ from $T$} is the operation $\psi(T, \Me) = T'$ where $T'\in \Te^{\Le'}$ is such that: 
	\begin{itemize}
		\item The topology of $T'$ is given by the TDR map applied to the internal edges in $T$; i.e., 
		$$\Se(T') = \Psi_{\Le'} \left(\Se(T)\right) = \left\{ \Psi_{\Le'}(s) \mid s \in \Se(T)\right\}.$$
		\item The length of each edge $p$ of $T'$ is the $L^2$ norm of the edges in the pre-image of $\Pe(T)$ under the TDR map, i.e.,
		$$|p|_{T'} := \sqrt{\sum_{q \in \Pe(T): \Psi_{\Le'}(q) = p} |q|^2_{T}}.$$ 
	\end{itemize}
\end{definition}

Pruning provides a precise method for removing leaves from a tree, and returns a unique, unambiguous tree. In contrast, adding leaves to a tree can be performed in several ways. Therefore, regrafting is not a function, and instead we must consider the set of trees produced by all possible regraftings. These sets, which we call sprouting spaces, play an important role in constructing distances in  Towering space. We refer to leaf prunings and leaf regraftings by the more general term \textit{leaf operations}.

\begin{definition}
	\label{Def:theSprouting}
	Given $T \in \Te^{\Le'}$ and a superset of leaves $\Le \supseteq \Le'$, \textbf{the sprouting space} of $T$ in $\Te^{\Le}$, denoted by $\Lambda^{\Le}(T)$, is the set of all the trees in $\Te^{\Le}$ from which $\Me = \Le \setminus \Le'$ are  prunable and that map to $T$ under the pruning of $\Me$. That is, 
		$\Lambda^{\Le}(T) = \left\{ t \in \Ze_{\Me}^{\Le} \mid \psi(t,\Me) = T \right\}$. 
\end{definition}

Sprouting spaces are closely related to the Extension spaces of \cite{GrindstaffGillian2019RoPL}: while Extension spaces are subsets of trees in a maximal BHV space that map onto the same tree under the TDR map (as originally defined in \citet{zairis2016genomic}), sprouting spaces are trees in a higher-dimensional BHV space that map onto the same tree after a pruning operation. Thus, we can consider sprouting spaces as a modified version of Extension spaces that permit traversal of BHV spaces of different dimension. 


\section{The Towering Space}
\label{sec:BuildingUp}

We now construct a space on a union of BHV spaces, and to study its properties. We begin by defining equivalence classes, which span trees in different BHV levels, and use them to construct a distance metric, thereby extending the BHV metric space $(\Te^{\Le}, d_{\text{BHV}})$ to trees with differing leaf sets. We then explore geometrical results arising from the properties of these equivalence classes. We conclude this section with a characterization of geodesics in  Towering space, and a preliminary algorithm to compute geodesics.

\subsection{Towering distances}

For a set of $n$ leaves $\Ne$, consider all possible subsets $\Le \subset \Ne$ with cardinality $|\Le| \geq 3$, and their corresponding BHV spaces $\Te^{\Le}$. Towering space contains all trees in these BHV spaces, that is, all trees in $\Te^{\mathbf{P}(\Ne)} = \bigcup_{\Le \in \mathbf{P}(\Ne)} \Te^{\Le}$, where $\mathbf{P}(\Ne) = \left\{\Le \subseteq \Ne \big| |\Le| \geq 3\right\}$. 
By applying leaf prunings and regraftings, we can traverse BHV levels through trees that we consider equivalent. 

\begin{definition}
	\label{Def:EquivalenceClasses}
	Two trees $T_1 \in \Te^{\Le_1}$ and $T_2 \in \Te^{\Le_2}$ belong to the same equivalence class, denoted $T_1 \simeq T_2$, if there is a finite number of prunings and regraftings that can be applied sequentially to $T_1$ that result in $T_2$. 
\end{definition}

Armed with this equivalence relation, we now construct the metric of Towering space by refining a preliminary metric on $\Te^{\mathbf{P}(\Ne)}$ to a strictly finite-valued quotient pseudometric \citep[I.5, Definition 5.19]{Bridson:1999ky}. The preliminary metric is as follows:
\begin{equation*}
    d^{*}(T_1, T_2) = \begin{cases}
        d_{\text{BHV}} (T_1, T_2) & \text{ if } \Le(T_1) = \Le(T_2),\\
        \infty & \text{ otherwise.}
    \end{cases}
\end{equation*}

\begin{definition}
	\label{Def:ToweringDistance}
	The \textbf{Towering distance} between trees $T_1, T_2 \in \Te^{\mathbf{P}(\Ne)}$ is the quotient pseudometric of $d^{*}$ on $\Te^{\mathbf{P}(\Ne)}/\simeq$, equivalently, the size of the smallest possible path between $T_1$ and $T_2$, allowing for a finite number of distance-zero leaf operations between trees in an equivalence class: 
	\begin{equation}
\label{Eq:ToweringDistance}
		d(T_1,T_2) = \inf\left\{\left.\sum_{i=1}^{k} d_{\text{BHV}}(t_i,t'_i) \right| T_1 \simeq t_1, t'_i \simeq t_{i+1} \forall i=1,\hdots, k-1, t'_k \simeq T_2\right\}, 
\end{equation} with the infimum taken over all finite sequences $\{t_i\}_{i=1}^{k}, \{t'_{i}\}_{i=1}^{k}$. 
\end{definition}

By definition, $d(\cdot, \cdot)$ is a pseudometric, upholding symmetry and the triangle inequality. Moreover, it correctly distinguishes non-deformed trees, meaning that if $T$ is a tree with strictly positive edge lengths, then  $d(T,T') = 0$ if and only if $T \simeq T'$. This in turn implies $T' \in \Lambda^{\Le'}(T)$ for some set of leaves $\Le'$. We can further refine $d$ into a proper metric function on $\Te^{\mathbf{P}(\Ne)}/\simeq$ by defining $T_1 \simeq T_2$ whenever $d(T_1,T_2) = 0$, thus guaranteeing that distinct trees are differentiated by the Towering space metric (see Supplementary Section 2). 



\subsection{Geometry of Towering geodesics}
\label{subSect:FirstMajorResults}
\label{sec:ShortPathsLowBHVspaces}

Our first result is a consequence of the reverse triangle inequality $||x|| - ||y|| \leq ||x-y||$ 
and the $L^2$ norm for combining edge lengths following leaf prunings. It implies that 
shorter paths cannot be obtained in higher BHV levels that involve leaves not originally present in the trees. This is a highly desirable property, as the addition of arbitrary extraneous leaves should not improve distances between trees. 

\begin{theorem}
	\label{thm:NoUpperShortcuts}
	Given two trees $T_1,T_2 \in \Te^{\Le'}$ for a subset of leaves $\Le' \subseteq \Le$, then $d_{\text{BHV}}(T_1, T_2) \leq d_{\text{BHV}}(T^{\uparrow}_1,T^{\uparrow}_2)$ for any trees $T^{\uparrow}_1 \in \Lambda^{\Le}(T_1)$ and $T^{\uparrow}_2 \in \Lambda^{\Le}(T_2).$
\end{theorem}

A consequence of this is that when a set of leaves is prunable from two trees, a shorter geodesic can be found in the lower BHV level resulting from this prune. 

\begin{corollary}
	\label{cor:M_trim_immediately}
	Given trees $T_1, T_2$ in the same $\Me$-trimmable space $\mathcal{Z}_{\Me}^{\Le}$, the length of the geodesic between them is lower-bounded by the length of the geodesic between their respective prunings of $\Me$; that is, $d_{\text{BHV}}(T_1,T_2) \geq d_{\text{BHV}}(\psi(T_1,\Me),\psi(T_2,\Me))$.
\end{corollary}




The construction of the short paths through low BHV levels will involve leaf operations performed at optimal locations. The following results describe the path length obtained through a single optimized prune-and-regraft operation (see Figure \ref{fig:ShortPaths}(a)). 

\begin{theorem}
	\label{Thm:First_GeodesicBest}
	Consider $T_1 \in \Te^{\Le}$ and $T_2 \in \Te^{\Le'}$ with $\Le = \Le' \sqcup \Me$. Define $T'_1 = \psi(T_1^{\perp \Me},\Me)$, the pruning of $\Me$ from the projection of $T_1$ onto $\Ze_{\Me}^{\Le}$. Then:
	\begin{enumerate}[i.]
		\item  The BHV distance between $T_1$ and $t \in \Lambda^{\Le}(T_2)$ is lower-bounded by the $L^2$ norm of the distances between $T_1$ and the $\Me$-trimmable subspace, and between $T_2$ and $T'_1$; that is, 
		\begin{equation}
		 \label{Eq:DirectDistance}
		d_{\text{BHV}}(T_1,t) \geq \sqrt{||P^{\downarrow \Me}(T_1)||^2 + d_{\text{BHV}}^2 (T'_1,T_2)} \text{ for every } t \in \Lambda^{\Le}(T_2)    	
		\end{equation}
		\item There exists at least one tree in $\Lambda^{\Le}(T_2)$ that achieves this bound. 
	\end{enumerate}
\end{theorem}

While this result describes an efficient way to go from $T_1$ to $T_2$ (through pruning a tree in $\Lambda^{\Le}(T_2)$), 
the following two results further establish that this path is in fact the most efficient overall when the only transitions between levels are prunings.

\begin{lemma}
	\label{lemma:DirectPruning}
	Given $T_1 \in \Te^{\Le}$ and $T_2 \in \Te^{\Le'}$ such that $\Le' \subset \Le$, consider all paths from $T_1$ to $T_2$ where the only transition between BHV spaces is the leaf pruning of $\Me = \Le \setminus \Le'$.
	The lengths of these paths are lower bounded by the length of the direct path from $T_1$ to $T_2$ by pruning $\Me$ at a tree in $\Lambda^{\Le}(T_2)$, with length given by \eqref{Eq:DirectDistance}.
\end{lemma}

\begin{lemma}
	\label{lemma:OnePruningRatherThanTwo}
	Consider $T_1 \in \Te^{\Le}$ and $T_2 \in \Te^{\Le'}$ where $\Le' \subset \Le$, and partition $\Me = \Le \setminus \Le'$ into $\Me = \Me_1 \sqcup \Me_2$. Consider all the paths from $T_1$ to $T_2$ where only two transitions are performed: pruning $\Me_1$, followed by pruning $\Me_2$ at some point later. All of these paths are at least as long as the direct path from $T_1$ to $T_2$, where $\Me$ is pruned in a tree in $\Lambda^{\Le}(T_2)$, with length given by \eqref{Eq:DirectDistance}. 
\end{lemma}

Lemmas \ref{lemma:DirectPruning} and \ref{lemma:OnePruningRatherThanTwo} show that any path between two trees $T_1$ and $T_2$ in $\Te^{\mathbf{P}(\Ne)}$ can be shortened by consolidating portions of the path involving only pruning leaves into a single pruning operation, allowing a direct jump to a tree at the lowest level of that portion. 
Similarly, sections of a path that exclusively regraft leaves can be combined into a single regrafting operation. This culminates in the following theorem relating to a prune-and-regraft operation.

\begin{theorem}
	\label{Thm:shorterPathsBelow}
	Consider two trees $T_1 \in \mathcal{T}^{\mathcal{L}_1}$ and $T_2 \in \mathcal{T}^{\mathcal{L}_2}$, with $\mathcal{L}_1 = \mathcal{L}' \sqcup \mathcal{M}_1$ and $\mathcal{L}_2 = \mathcal{L}' \sqcup \mathcal{M}_2$ for some subset $\mathcal{L}' \subseteq \mathcal{L}_1 \cap \mathcal{L}_2$. Among the paths from $T_1$ to $T_2$ where the leaves in $\mathcal{M}_1$ are all pruned, followed by regraftings of leaves in $\mathcal{M}_2$, the shortest paths have length
	\begin{equation}
		\label{Eq:ShortestLowerPath}
		\sqrt{\left[||P^{\downarrow \mathcal{M}_1}(T_1)|| + ||P^{\downarrow \mathcal{M}_2}(T_2)|| \right]^2 + d^2_{\text{BHV}}(T'_1,T'_2)},
	\end{equation}
	where $T'_i$ is the pruning of $\mathcal{M}_i$ from $T_i^{\perp \mathcal{M}_i}$, for $i = 1,2$. 
\end{theorem}

Therefore, the shortest path through a common lower level depends entirely on the $L^2$ norm of the edges that prevent leaf sets from being prunable in each tree, along with the BHV distance between the two trees at the lower level resulting from pruning the relevant leaves. A consequence of Theorem \ref{Thm:First_GeodesicBest} is that the shortest path between $T_1$ and $T_2$ via one common lower level is comprised of only two BHV geodesics, one completely contained in the BHV space of $T_1$ and the second completely contained in the BHV space of $T_2$  (see Figure \ref{fig:ShortPathsB}). These path segments are connected by pruning all leaves in $\mathcal{M}_1$ from a tree in the sprouting space $\Lambda^{\Le_1}(X)$ for some tree $X \in \Te^{\Le'}$ and regrafting all leaves in $\mathcal{M}_2$ onto the same tree to obtain a tree in $\Lambda^{\Le_2}(X)$. Moreover, the optimal $X \in \mathcal{T}^{\Le'}$ falls on the BHV geodesic from $T'_1$ to $T'_2$, and it is the tree located at the position proportional to the distances from the original trees $T_1$ and $T_2$ to the respective prunable spaces. Specifically, $X$ satisfies $$d_{\text{BHV}}(T'_1,X) = \frac{||P^{\downarrow \Me_1}(T_1)||}{||P^{\downarrow \Me_1}(T_1)|| + ||P^{\downarrow \Me_2}(T_2)||} d_{\text{BHV}}(T'_1,T'_2).$$ We provide a detailed proof and a method for finding $X$ in Supplementary Section 1. 

\begin{figure}
    \centering
    \subfigure[]{\includegraphics[width = 0.43\textwidth]{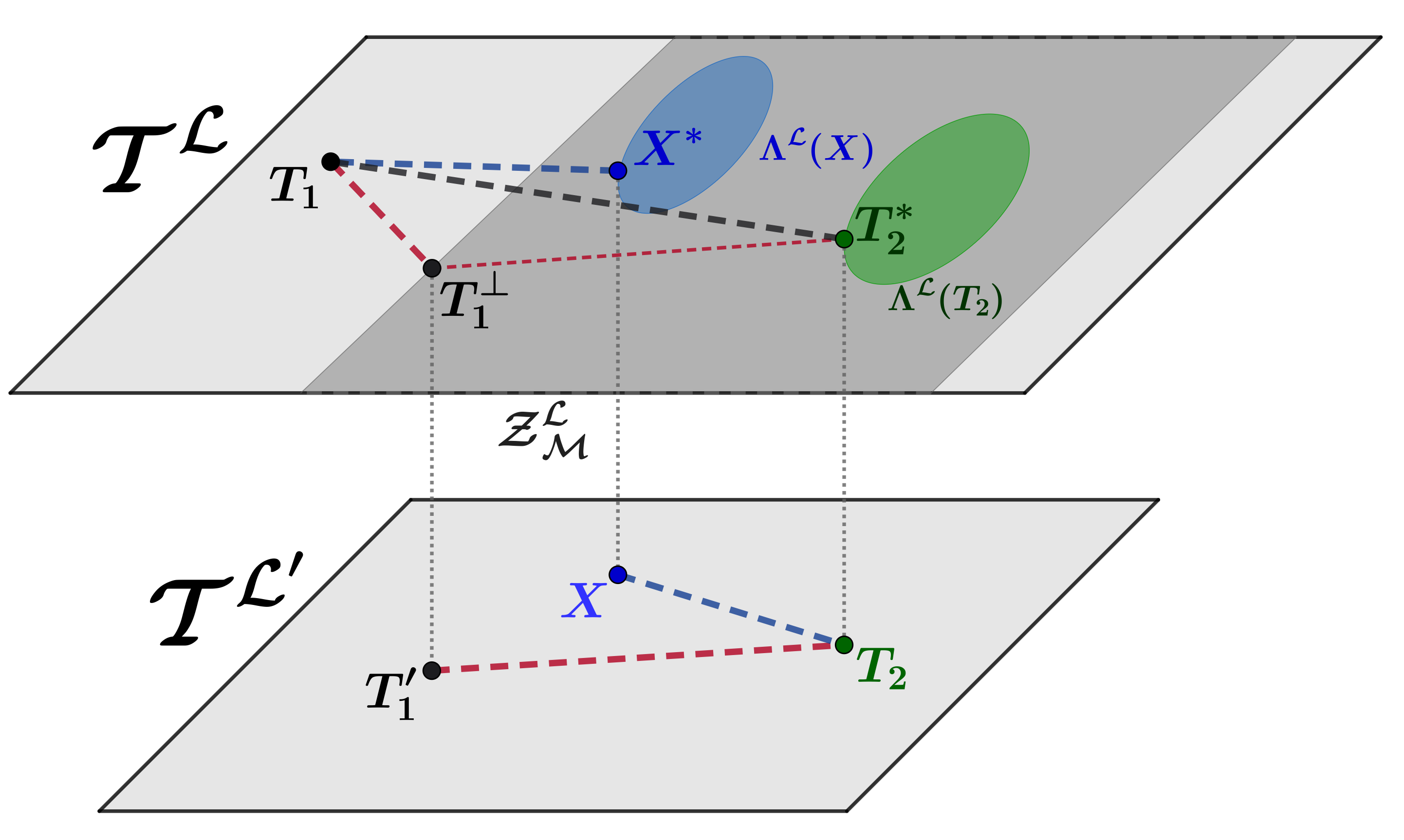}\label{fig:ShortPathsA}} 
    \subfigure[]{\includegraphics[width = 0.56\textwidth]{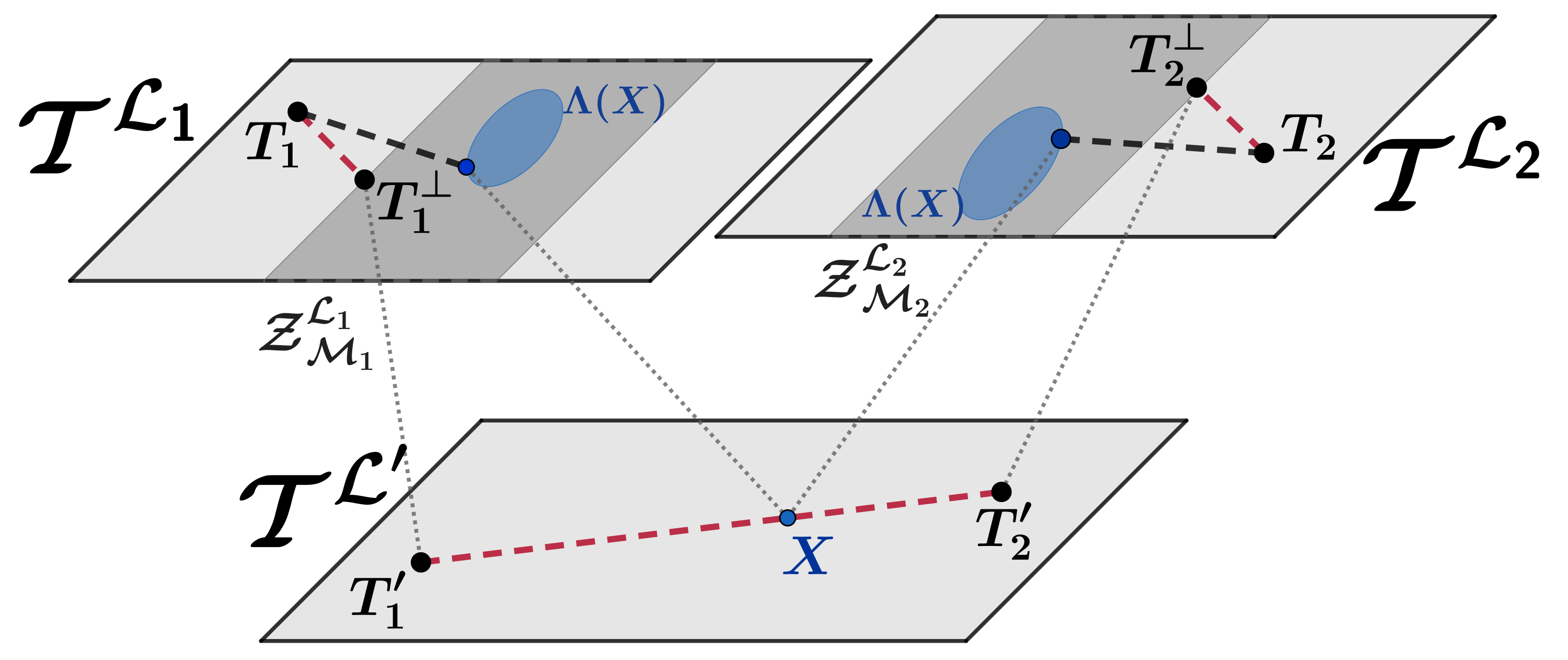}\label{fig:ShortPathsB}}
    \caption{Shortest paths via a single pruning and regrafting through a lower level. Shortest paths are denoted by black dashed lines. 
    (a) The shortest path from $T_1 \in \mathcal{T}^{\Le}$ to $T_2 \in \Te^{\Le'}$ through a single pruning of $\Me = \Le \setminus \Le'$ is through a tree $T^{*}_2 \in \Lambda^{\Le}(T_2)$ (green shading). 
    The length of this path is the $L^2$ norm of the distance from $T_1$ to its projection $T_1^{\perp}\in \Ze^{\Le}_{\Me}$ and the distance from the pruned tree $T'_1$ to $T_2$ (the hypotenuse of the ``triangle'' shown by red dashed lines). 
    Any path through a different leaf pruning (blue dashed line) is longer than the path through $T^{*}_2$. 
    (b) The shortest path from $T_1 \in \Te^{\Le_1}$ to $T_2 \in \Te^{\Le_2}$ through a single prune-and-regraft operation traverses a top level space to reach $\Lambda^{\mathcal{L}_1}(X)$ (blue shading) for an optimal tree $X \in \Te^{\Le'}$, then prunes and immediately regrafts to $\Lambda^{\mathcal{L}_2}(X)$ (blue shading), before again traversing the top level space.
    $T_1^{\perp}$ and $T_2^{\perp}$ represent the projections of the endpoint trees to the trimmable spaces (gray shading), and the optimal tree $X$ is on the geodesic between the prunings of these trees ($T'_1$ and $T'_2$). While the path through these projections (red dashed lines) is not efficient, the length of its segments are combined to obtain the length of the optimal path (Theorem \ref{Thm:First_GeodesicBest}).} 
    \label{fig:ShortPaths}
\end{figure}


\begin{figure}
    \centering
    \subfigure[]{\includegraphics[width = 0.9\textwidth]{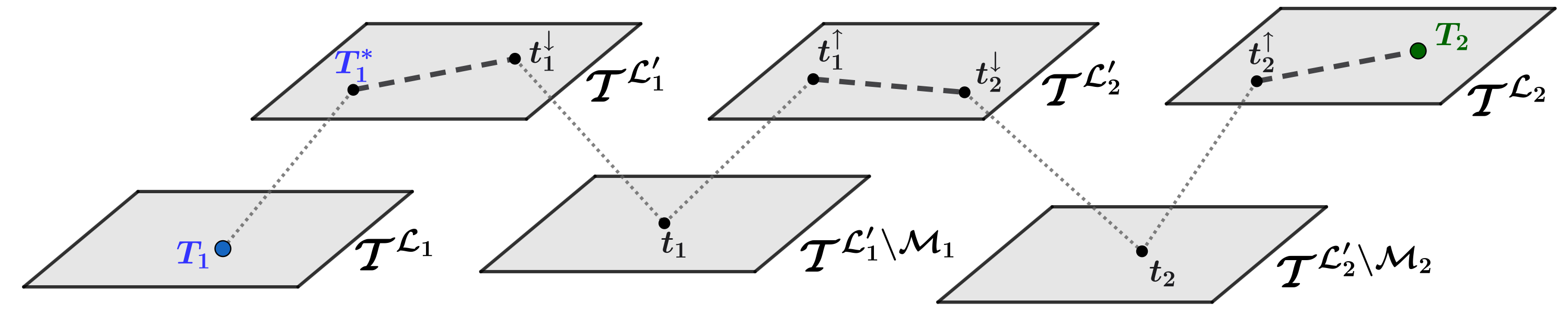}\label{fig:JumpingPathA}}
    \subfigure[]{\includegraphics[width = 0.9\textwidth]{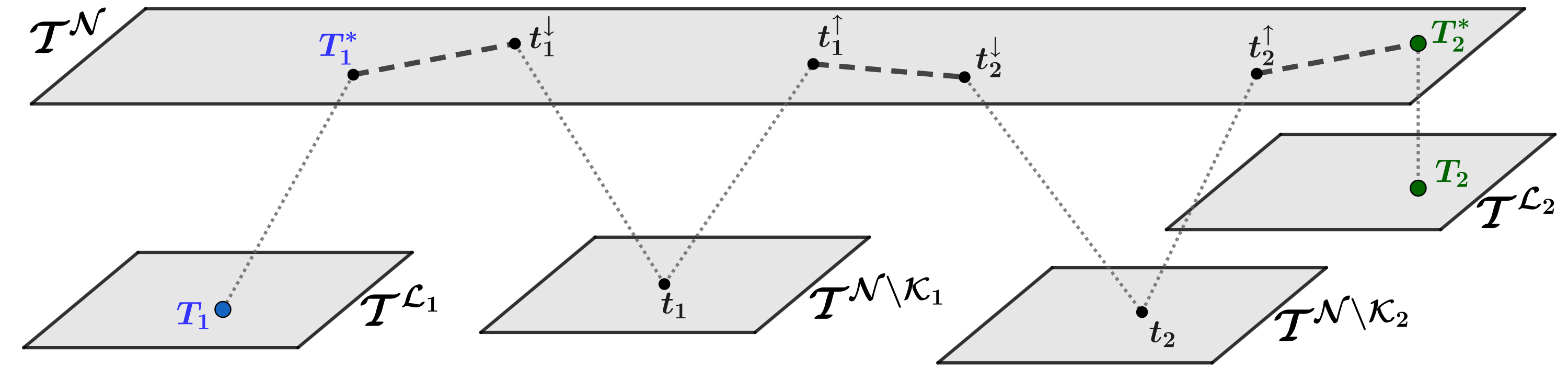}\label{fig:JumpingPathB}}
    \caption{Efficient paths from $T_1 \in \Te^{\Le_1}$ to $T_2 \in \Te^{\Le_2}$. (a) The path length from $T_1$ to $T_2$ can be reduced by consolidating consecutive prunings into single operations, and likewise for regraftings. Different BHV spaces on this path are connected through one prune-and-regraft operation at a time. (b) The path length is minimized by mapping all geodesic segments to the common upper level $\Te^{\Ne}$. This path passes through trees $T^{*}_1 \in \Lambda^{\Ne}(T_1)$ and $T^{*}_2 \in \Lambda^{\Ne}(T_2)$ by a sequence of segments in $\Te^{\Ne}$ that are connected only by prune-and-regraft operations. Notably, these paths do not traverse positive distances through the lower BHV levels.} 
    \label{fig:JumpingPath}
\end{figure}

\subsection{Characterizing geodesics}
\label{sec:ShortPathsHighBHVspaces}



While Theorem \ref{Thm:shorterPathsBelow} provides insight into the geometry of paths in Towering space, 
in practice, geodesics will often involve repeated prunings and regraftings. From Lemmas \ref{lemma:DirectPruning} and \ref{lemma:OnePruningRatherThanTwo}, we know that for each transition between BHV levels, a candidate geodesic path may pass through up to three trees within the same equivalence class --- one prior to a leaf pruning, one following a leaf regrafting, and one in the lowest BHV level. We can therefore describe a path from $T_1$ to $T_2$ by considering trees $t_1, \hdots, t_r$ where transitions between BHV levels occur. For each $i \in \{1, \hdots, r\}$, we denote by $t_i^{\downarrow}$ the tree prior to a pruning, and by $t_i^{\uparrow}$ the tree following the corresponding regrafting (see Figure \ref{fig:JumpingPathA}).

We now refine the structure of these candidate paths and show that at least one geodesic between any pair of trees consists of segments fully contained in the upper level $\mathcal{T}^{\Ne}$. We refine further by applying Theorem \ref{Thm:shorterPathsBelow} to restrict which trees allow optimal connections between the segments, thereby characterizing Towering space geodesics. We proceed first by demonstrating that BHV geodesics can be matched in length by geodesics in higher BHV levels. This result allows paths between any pair of trees trees to be lifted to an equally long path in the upper level $\Te^{\Ne}$. 

\begin{lemma}
    \label{lemma:LiftingGeodesics}
    Consider $T'_1,T'_2 \in \mathcal{T}^{\Le'}$ and a superset of their leaves $\Le \supset \Le'$. There exists $T_1 \in \Lambda^{\Le}(T'_1)$ and $T_2 \in \Lambda^{\Le}(T'_2)$ such that $d_{\text{BHV}}(T_1,T_2) = d_{\text{BHV}}(T'_1,T'_2)$.     
\end{lemma}

Thus, any path can be mapped to a path where all positive-length segments are contained in the uppermost BHV level $\Te^{\Ne}$ (Figure \ref{fig:JumpingPathB}). Therefore, in order to find a geodesic, it is sufficient to find the path with minimum length among all paths of this shape. This result is summarized in the following theorem. 

\begin{theorem} \label{Thm:4}
	Given $T_1 \in \Te^{\Le_1}$ and $T_2 \in \Te^{\Le_2}$, consider the set of common leaves $\Le' = \Le_1 \cap \Le_2$. A geodesic between $T_1$ and $T_2$ passes through trees $T^{*}_1 = t^{\uparrow}_{0}$, $t^{\downarrow}_1$, $t^{\uparrow}_1$, ..., $t^{\downarrow}_r$, $t^{\uparrow}_r$, $t^{\downarrow}_{r+1} = T^{*}_2 \in \Te^{\Le}$, where
	\begin{itemize}
		\item $T^{*}_1 \in \Lambda^{\Ne}(T_1)$ and $T^{*}_2 \in \Lambda^{\Ne}(T_2)$, and
		\item For each $i \in \{1,\hdots,r\}$, $t^{\downarrow}_i \simeq t^{\uparrow}_i$ with $t^{\downarrow}_i, t^{\uparrow}_i \in \Lambda^{\Ne}(t_{i})$ for some tree $t_i \in \Te^{\Ne \setminus \Ke_{i}},$
	\end{itemize}
	for a partition $\Le' = \Ke \sqcup \Ke_1 \sqcup \hdots \Ke_{r}$. 
	The length of the geodesic is given by $$\sum_{i=1}^{r+1} d_{\text{BHV}}(t^{\uparrow}_{i-1},t^{\downarrow}_{i}).$$ 
\end{theorem}

Furthermore, Theorem \ref{Thm:shorterPathsBelow} constrains our search for $t_1, \hdots, t_r$. Specifically, fixing the trees $t_{i-1}^{\uparrow}$ and $t_{i+1}^{\downarrow}$, we know $t_i$ is on the BHV geodesic from $\psi(t_{i-1}^{\uparrow \perp \Ke_i}, \Ke_i)$ to $\psi(t_{i+1}^{\downarrow \perp \Ke_i},\Ke_i)$. This indicates that along the path from tree $t^{\uparrow}_{i-1}$ to $t^{\downarrow}_{i+1}$, the edges outside $P^{\downarrow \Ke_i}(t^{\uparrow}_{i-1})$ and $P^{\downarrow \Ke_i}(t^{\downarrow}_{i+1})$ change at the same rate as their counterparts along the geodesic at the lower level. This implies the rate of change of edges remains the same along the path from one BHV space to the next, until that edge is (possibly) dropped. Thus the path is a constant-rate morph of one tree into the other. 

From this, and the optimized expression shown in Theorem \ref{Thm:shorterPathsBelow}, we conclude the length of a geodesic will directly depend on the optimal partition of common leaves to be pruned and regrafted through the edges preventing these leaves from being pruned, similar to the way $P^{\downarrow\Me_i}(T_i)$ contribute to the expression in \eqref{Eq:ShortestLowerPath}. More explicitly, if we fix the endpoint trees $T^{*}_1 \in \Lambda^{\Ne}(T_1)$, $T^{*}_2 \in \Lambda^{\Ne}(T_2)$ and the partition $\Le' = \Ke \sqcup \Ke_1 \sqcup \hdots \sqcup \Ke_r$, the length of the shortest path through the endpoint trees 
is 
\begin{equation}
	\label{Eq:GeodesicExpression}
	\sqrt{\sum_{i=1}^{r} \left[P^{*\downarrow \Ke_i}(T_1) + P^{*\downarrow \Ke_i}(T_2)\right]^2 + d^2_{\text{BHV}}(T'_1, T'_2)},
\end{equation}
where $P^{*\downarrow \Ke_i}(T_1)$ and $P^{*\downarrow \Ke_i}(T_2)$ are modified versions of the sets of edges preventing $\Ke_i$ from being prunable, and $T'_1$ and $T'_2$ are the underlying trees resulting from pruning all leaves $\Le'\setminus \Ke$ after reducing all edges in $\bigcup_{i=1}^{r} P^{*\downarrow \Ke_i}(T_1)$ and $\bigcup_{i=1}^{r} P^{*\downarrow \Ke_i}(T_2)$ to zero. Note $P^{*\downarrow \Ke_i}(T_1)$ and $P^{*\downarrow \Ke_i}(T_2)$ may differ slightly from $P^{\downarrow \Ke_i}(T_1)$ and $P^{\downarrow \Ke_i}(T_2)$ depending on where other common leaves are regrafted. 

\subsection{An algorithm to compute geodesics} \label{subsec:specific_alg} 

Based on the above results, we propose an iterative algorithm for computing geodesics in Towering space. The algorithm consists of two main steps: initializing candidate paths between $T_1$ and $T_2$; and optimizing these paths using Theorems \ref{Thm:First_GeodesicBest} and \ref{Thm:shorterPathsBelow}. We begin by detailing the optimization step. 

Choose any partition $\Le' = \Ke \sqcup \Ke_1 \sqcup \cdots \sqcup \Ke_r$. Initialize a candidate path with arbitrary initial trees $T^{*}_1(0) \in \Lambda^{\mathcal{N}}(T_1)$ and $T^{*}_2(0) \in \Lambda^{\mathcal{N}}(T_2)$, and transition trees $t^{\downarrow}_{i}(0)$, $t_{i}(0)$, $t^{\uparrow}_{i}(0)$, where $t_{i}(0) \in \mathcal{N} \setminus \mathcal{K}_i$ and $t^{\downarrow}_{i}(0), t^{\uparrow}_{i}(0) \in \Lambda^{\mathcal{N}}(t_{i}(0))$. The algorithm iteratively improves the path length by updating the transition trees. For iteration $k$:

\begin{itemize}
    \item \textbf{If $k$ is odd} (forward pass from $T_1$ to $T_2$):
    \begin{itemize}
        \item Apply Theorem \ref{Thm:First_GeodesicBest} to update the initial tree $t^{\uparrow}_{0}(k) = T^{*}_1(k) \in \Lambda^{\mathcal{N}}(T_1)$, selecting the closest tree to $t^{\downarrow}_{1}(k-1)$.
        \item For each $j = 1, \hdots, r$, apply Theorem \ref{Thm:shorterPathsBelow} to optimize the path from $t^{\uparrow}_{j-1}(k)$ to $t^{\downarrow}_{j+1}(k-1)$ via pruning and regrafting $\mathcal{K}_j$, obtaining updated trees $t^{\downarrow}_{j}(k)$, $t_{j}(k)$, and $t^{\uparrow}_{j}(k)$.
    \end{itemize}
    
    \item \textbf{If $k$ is even} (backward pass from $T_2$ to $T_1$):
    \begin{itemize}
        \item Apply Theorem \ref{Thm:First_GeodesicBest} to update the final tree $t^{\downarrow}_{r+1}(k) = T^{*}_2(k) \in \Lambda^{\mathcal{N}}(T_2)$, selecting the closest tree to $t^{\uparrow}_{r}(k-1)$.
        \item For each $j = r, \dots, 1$, apply Theorem  \ref{Thm:shorterPathsBelow} to optimize the path from $t^{\downarrow}_{j+1}(k)$ to $t^{\uparrow}_{j-1}(k-1)$ via pruning and regrafting $\mathcal{K}_j$, obtaining updated trees $t^{\downarrow}_{j}(k)$, $t_{j}(k)$, and $t^{\uparrow}_{j}(k)$.
    \end{itemize}
 
    \item Repeat until no further improvements in the path length are observed.
\end{itemize}

The above algorithm conditions on a fixed partition $\Le' = \Ke \sqcup \Ke_1 \sqcup \cdots \sqcup \Ke_r$. To ensure convergence to a geodesic (rather than a local mimimum), as a preliminary approach, we suggest constructing a candidate path for each possible partition. Choosing appropriate partitions and candidate paths can be done as follows:
\begin{itemize}
    \item For each possible topology in $\Lambda^{\mathcal{N}}(T_1)$ and $\Lambda^{\mathcal{N}}(T_2)$, choose initial trees $T^{*}_1(0) \in \Lambda^{\mathcal{N}}(T_1)$ and $T^{*}_2(0) \in \Lambda^{\mathcal{N}}(T_2)$ .
    \item For each $j = 1, \dots, r$:
    \begin{itemize}
        \item Project $t^{\uparrow}_{j-1}(0)$ onto the trimmable space $\mathcal{Z}^{\mathcal{N}}_{\mathcal{K}_j}$, and set $t_{j}(0) = \psi\left(t^{\uparrow \perp \mathcal{K}_j}_{j-1}(0), \mathcal{K}_j\right)$. 
        \item For each possible topology in $\Lambda^{\mathcal{N}}(t_{j}(0))$, choose a representative regrafting $t^{\uparrow}_{j}(0)$.
    \end{itemize}
\end{itemize}

By the completeness of Towering space, Theorem \ref{Thm:shorterPathsBelow} and Theorem \ref{Thm:4}, this procedure is guaranteed to find at least one geodesic. While this algorithm is guaranteed to find a solution, it is not efficient. We leave further refinements --- especially in selecting efficient partitions of $\Le'$ to reduce the number of candidate paths --- to future work.

\section{Discussion} \label{sec:discussion} 

The primary contribution of this paper is the introduction of Towering space, a complete metric space that is defined over \textit{all} possible phylogenetic trees on a collection of leaves, regardless of whether their leaf sets are identical. This framework enables comparisons between any pair of trees, 
laying a foundation for extending data-analytic tools to fully general phylogenetic comparisons. 

Towering space shares many desirable geometric properties of BHV space, including path traversal via smooth topological and edge-length deformations. This was made possible through two subtle but important innovations: ensuring that internal edges map to internal edges under prunings (Definition \ref{Def:prunable_All}), and using the $L^2$ norm for merging edge-lengths (Definition \ref{Def:LeafPrune}). Edge-length merging via $L^2$ has a particularly interesting interpretation. 
In BHV space, the origin tree is identical for every topology, corresponding to the tree with all branches of zero length. This tree represents no structured evolutionary divergence among the organisms labeling its leaves. The BHV distance of a tree to the origin represents the cumulative divergence represented by all the branches in the tree, thus, trees further from the origin reflect evolutionary processes with more substantial changes than those closer to the origin. The use of the $L^2$ norm to combine branch lengths in leaf prunings makes the trees connected by them equidistant from the origin (in their respective BHV space): $d_{BHV}(0, T) = d_{BHV}(0, \psi(T,\Me))$ for any $T \in \Ze^{\Le}_{\Me}$, where $0$ denotes the origin tree. Thus, a leaf pruning produces a new tree with a similar topology and reflecting a comparable amount of overall evolutionary change among the organisms in the leaf set. This is an appealing biological interpretation of distances in Towering space. 
Furthermore, the BHV distance between $T_1, T_2 \in \Te^{\Le}$ can be directly expressed in terms of the distance of each tree to the origin tree and differences among their topologies \citep[Section 4.2]{BILLERA2001} by 
\begin{equation}
\label{eq:BHVLinkOrigin}
    d^2_{\text{BHV}}(T_1, T_2) = ||T_1||^2 + ||T_2||^2 - 2||T_1|| ||T_2|| \cos \left[\min\{\pi, \measuredangle (T_1,T_2)\}\right],
\end{equation}
where $||T|| = d_{\text{BHV}}(0,T)$ is the distance from a tree to the origin tree and $\measuredangle (T_1,T_2)$ is the angle between $T_1$ and $T_2$ via the origin tree. This angle is primarily a reflection of differences in their topologies. 
Since all trees in a sprouting space $\Lambda^{\Le}(T)$ for a given tree $T \in \Te^{\Le'}$ with $\Le'\subset \Le$ are equidistant from the origin tree, \eqref{eq:BHVLinkOrigin} implies that the closest tree in the sprouting space to any other tree $T' \in  \Te^{\Le}$ will mainly depend on the similarity of their topologies. This concurs with the intuition that, when comparing trees with different leaf sets, trees with similar topologies should lie close to each other --- for example, because they share an evolutionary history, but with some organisms unobserved. 

We also showed that the BHV distance between two trees after pruning gives a lower bound on the Towering distance between their pre-images (Theorem \ref{thm:NoUpperShortcuts}). This monotonicity property ensures that our metric remains coherent across different BHV levels by guaranteeing that the distance between trees is only as close as the most favorable distance in upper levels, after accounting for the loss of information through pruning. 
In addition, in Towering space, edges preventing leaf prunings are grouped and their contribution to the path length is identical to that of uncommon edges through support pairs in BHV space. Once such edges are accounted for, the remaining path proceeds through a lower-dimensional BHV space, where the geodesic length behaves identically to BHV distances.

Our results characterizing geodesics led naturally to an algorithm for distance computation. While the proposed algorithm will provably find a geodesic, it is a preliminary algorithm, and future work is required before the distance can be adopted for large-scale use. The structural similarity with BHV geodesics suggests the potential for adapting \cite{OwenMegan2011}'s algorithm to efficiently classify edges into those preventing prunings and those that belong to the underlying lower trees. This would give the efficient partition of $\Le'$. We see this as the most promising direction to pursuing a faster algorithm. We leave further algorithmic development to future work. 

Another key feature of Towering space is that it is a \textit{complete length space}, meaning all tree pairs are connected through continuous paths. However, in contrast to BHV space, 
Towering space does not guarantee unique geodesic paths. We conjecture that geodesics will be unique almost surely (e.g., for trees drawn from an absolutely continuous measure with full support on $\Te^{\mathbf{P}(\Ne)}$), and believe that the completeness and path-connectedness of the space will facilitate the development of analytical tools in Towering space (such as Fr\'echet mean and principal path construction). 

From a computational standpoint, Towering space also addresses limitations of BHV space generalizations that embedded trees in subspaces of the highest-dimension BHV space \citep{ren2017combinatorial, GrindstaffGillian2019RoPL}. While BHV distances can be minimized between pairs of these subspaces, the resulting solution is a dissimilarity measure, not a true metric, and the fastest algorithms for computing the dissimilarity are computationally expensive \citep{VW_Extensions}. In contrast, Towering distance is a true metric, and avoids high-dimensional embeddings by transitioning between BHV spaces of different dimensions through pruning and regrafting operations. While no computational implementation yet exists for Towering distances, we conjecture that an implementation of even the preliminary algorithm given in Section \ref{subsec:specific_alg} will be faster than the current fastest algorithms for computing distances between Extension spaces. Both approaches require outer loops over all topologies in $\Lambda^{\mathcal{N}}(T_1)$ and $\Lambda^{\mathcal{N}}(T_2)$, but Towering distances do not require computing gradients, and each local improvement is done through a closed-form operation (Theorem \ref{Thm:shorterPathsBelow}) instead of an approximate solution. This dramatically accelerates computation.

Finally, from a biological standpoint, Towering space geodesics provide a natural interpretation for gene gain and loss: pruning corresponds to the loss of a gene by a leaf organism (i.e., genome streamlining), while regrafting corresponds to the gain of of gene (e.g., via horizontal gene transfer). By formalizing these transformations within a metric structure, we advance tools to study evolution in organisms that frequently experience these events. 
Trees with non-identical leaf sets pose a particular challenge for reconstructing ancient evolutionary relationships, and we anticipate that the greatest potential benefit of Towering space is in the study of the emergence of complex life from its single-celled origins. 

\section*{Acknowledgements}

This work was supported by NIH NIGMS R35 GM133420 and NSF 2415614. 

\newpage

\section{Supplementary Material: Proofs of key results}

\begin{proof}[of Lemma \ref{lemma:Trimmable_convex}]

    \textbf{(Closure)} Consider $T \in \mathcal{T}^{\mathcal{L}}$ such that $T \notin \mathcal{Z}_{\mathcal{M}}^{\mathcal{L}}$. Since $\mathcal{M}$ is not mutually prunable from $T$ then either some external edge $e$ to a leaf in $\mathcal{M}$ has a positive length, or there is an internal split $s$ such that $\Psi_{\mathcal{L}\setminus \mathcal{M}}(s)$ is not an internal split on the leaves $\mathcal{L} \setminus \mathcal{M}$. In the first case, given a value $\epsilon < |e|_{T}$, is enough to guarantee $|e|_{T'}>0$ for any tree $T'$ such that $d_{\text{BHV}}(T,T') < \epsilon$, which implies $T' \notin \mathcal{Z}_{\mathcal{M}}^{\mathcal{L}}$. In the second case, if $\epsilon < |s|_T$ then similarly $|s|_{T'}>0$ (so that $s \in \mathcal{S}(T')$) for any tree $T'$ such that $d_{\text{BHV}}(T,T')< \epsilon$, implying $T'\notin \mathcal{Z}_{\mathcal{M}}^{\mathcal{L}}$. Thus, there is always a neighborhood of $T$ not intersecting $\mathcal{Z}_{\mathcal{M}}^{\mathcal{L}}$. Therefore the subspace is closed. 
	
	\textbf{(Convexity)} Consider two trees $T_1, T_2 \in \mathcal{Z}_{\mathcal{M}}^{\mathcal{L}}$ and $T$ on the geodesic from $T_1$ to $T_2$. Along the geodesic, the lengths of external edges change gradually from the lengths they have in $T_1$ to the lengths they have in $T_2$. Since all external edges to leaves in $\mathcal{M}$ are of length zero in both trees, then these external edges are also of length zero in $T$. Given any internal split $s$ of the leaves $\mathcal{L}$ such that $\Psi_{\mathcal{L}\setminus \mathcal{M}}(s)$ does not map to an internal split,  $s$ is not a part of the internal edges of $T_1$ nor the internal edges of $T_2$, which implies it is not in $\mathcal{S}(T)$. Thus,  $\mathcal{M}$ is mutually prunable from $T$, and $\mathcal{Z}_{\mathcal{M}}^{\mathcal{L}}$ is convex.
\end{proof}

\begin{proof}[of Lemma \ref{lemma:ProjectionCharacterize}]

    By definition, all external edges to leaves in $\mathcal{M}$ in $T^{\perp \mathcal{M}}$ are of size zero, and every internal edge maps to an internal split under the TDR map $\Psi_{\mathcal{L}\setminus \mathcal{M}}$, implying $T^{\perp \mathcal{M}} \in \mathcal{Z}_{\mathcal{M}}^{\mathcal{L}}$. It is straightforward to see $d_{\text{BHV}}(T, T^{\perp \mathcal{M}}) = ||P^{\downarrow \mathcal{M}}(T)||_{T}$, since all edges in $P^{\downarrow \mathcal{M}}$ are of size zero (or not present) in $T^{\perp \mathcal{M}}$ and those are the only edges where the two trees differ. Finally, for any tree $t'\in \mathcal{T}^{\mathcal{L}}$ such that $d_{\text{BHV}}(T,t') < ||P^{\downarrow \mathcal{M}}(T)||_{T}$ at least one of the edges in $P^{\downarrow \mathcal{M}}(T)$ is of positive length, thus not belonging to the $\mathcal{M}$-trimmable subspace.
    
\end{proof}


\begin{proof}[ of Theorem \ref{thm:NoUpperShortcuts}]
	Given two trees $T^{\uparrow}_1 \in \Lambda^{\Le}(T_1)$ and $T^{\uparrow}_2 \in \Lambda^{\Le}(T_2)$, we know both $T^{\uparrow}_1, T^{\uparrow}_2 \in \Ze^{\Le}_{\Me}$ for $\Me = \Le \setminus \Le$. Since $\Ze^{\Le}_{\Me}$ is convex and closed, the geodesic from $T^{\uparrow}_1$ to $T^{\uparrow}_2$ is fully contained in $\Ze^{\Le}_{\Me}$. Let $t^{\uparrow}_{1},\hdots, t^{\uparrow}_{k}$ be the trees along this geodesic at topology orthant boundaries, and take $T^{\uparrow}_1 = t^{\uparrow}_{0}$ and $ t^{\uparrow}_{k+1} = T^{\uparrow}_2$, so that the segment going from $t^{\uparrow}_{i-1}$ to $t^{\uparrow}_{i}$ is fully contained in a single topology orthant for $i = 1, \hdots, k+1$. 
	For each of these trees, take $t_{i} = \psi(t^{\uparrow}_i, \Me)$. The segment from each $t_{i-1}$ to $t_{i}$ is also fully contained in a single orthant. Moreover, the reverse triangle inequality implies $d_{\text{BHV}}(t_{i-1},t_{i}) \leq d_{\text{BHV}}(t^{\uparrow}_{i-1},t^{\uparrow}_{i})$. Since the lower trees describe a valid path from $T_1$ to $T_2$, and it is already a shorter path than the geodesic in the upper BHV level, we have that $d_{\text{BHV}}(T_{1},T_{2}) \leq d_{\text{BHV}}(T^{\uparrow}_{1},T^{\uparrow}_{2})$.
	
\end{proof}

To construct shortest paths in Towering Space, leaf operations along these paths should be performed in optimal positions. To select points where prunings and regraftings take place, we introduce a structure that divides the edges of a tree $T$ in $\Pe^{\downarrow \Me}(T)$ (those edges preventing $\Me$ from being prunable from $T$) into separate independent subtrees. This division plays an important role in finding the best trees described in Theorems 2 and 3, and will later be used in distance computation. 

\begin{definition}
	\label{Def:independentSets}
	Given a tree $T \in \mathcal{T}^{\mathcal{L}}$ and a subset of leaves $\mathcal{M} \subset \mathcal{L}$, consider the natural partition of $\mathcal{M}$ into subsets that group into a single node when $T$ is projected onto the $\mathcal{M}$-trimmable space. This can be achieved by a partition $\mathcal{M} = \mathcal{M}_1 \sqcup \hdots \sqcup \mathcal{M}_{r}$, called \textbf{the independent maximal sets} of $\mathcal{M}$ in $T$, such that for each $\mathcal{M}_i$:
	
	 \begin{enumerate}
	 	\item There is a leaf $\ell \in \mathcal{L} \setminus \mathcal{M}$ such that $s_i = \mathcal{M}_{i} \cup \{\ell\} \big| \mathcal{L} \setminus \left(\mathcal{M}_{i} \cup \{\ell\}\right) \in \mathcal{S}(T)$ and there is no $\mathcal{M}'\subseteq \mathcal{M}$ such that $\mathcal{M}_i \subset \mathcal{M}'$ and $\mathcal{M}' \cup \{\ell\} \big| \mathcal{L} \setminus \left(\mathcal{M}' \cup \{\ell\}\right) \in \mathcal{S}(T)$. In this case, $\mathcal{M}_i$ is  \textbf{the independent maximal set neighboring $\ell$}; or 
	 	\item  $s_i = \mathcal{M}_{i} \big| \mathcal{L} \setminus \mathcal{M}_{i} \in \mathcal{P}(T)$ and there is no $\mathcal{M}'\subseteq \mathcal{M}$ such that $\mathcal{M}_i \subset \mathcal{M}'$ and either $\mathcal{M}' \big| \mathcal{L} \setminus \mathcal{M}' \in \mathcal{S}(T)$ or $\mathcal{M}' \cup \{\ell\} \big| \mathcal{L} \setminus \left(\mathcal{M}' \cup \{\ell\}\right) \in \mathcal{S}(T)$ for some leaf $\ell \in \mathcal{L}\setminus \mathcal{M}$. We refer to $\mathcal{M}_i$ as \textbf{an independent maximal set attached to internal edges}, since the edge $s_i$ must be adjacent to two internal edges in $T$ or incident on a node with degree higher than 3. 
	 \end{enumerate}
	 
	 Given an independent maximal set $\mathcal{M}_i$, consider the set
	 $$P^{\downarrow \mathcal{M}}_i(T) = \left\{s \in \mathcal{S}(T) \mid s = \mathcal{M}' \big| \mathcal{L}\setminus \mathcal{M}' \text{ or } \mathcal{M}' \cup \{\ell\} \big| \mathcal{L}\setminus (\mathcal{M}' \cup \{\ell\}) \text{ for } \mathcal{M}'\subseteq \mathcal{M}_i \right\}$$
	 which we call \textbf{the edges in $T$ belonging to} $\mathcal{M}_i$. Note $P^{\downarrow \mathcal{M}}_i(T) \subset P^{\downarrow \mathcal{M}}(T)$ (Figure \ref{fig:Example_IndependentSubTrees}).
\end{definition}

\begin{figure}
	\centering
	\includegraphics[width = 0.60\textwidth]{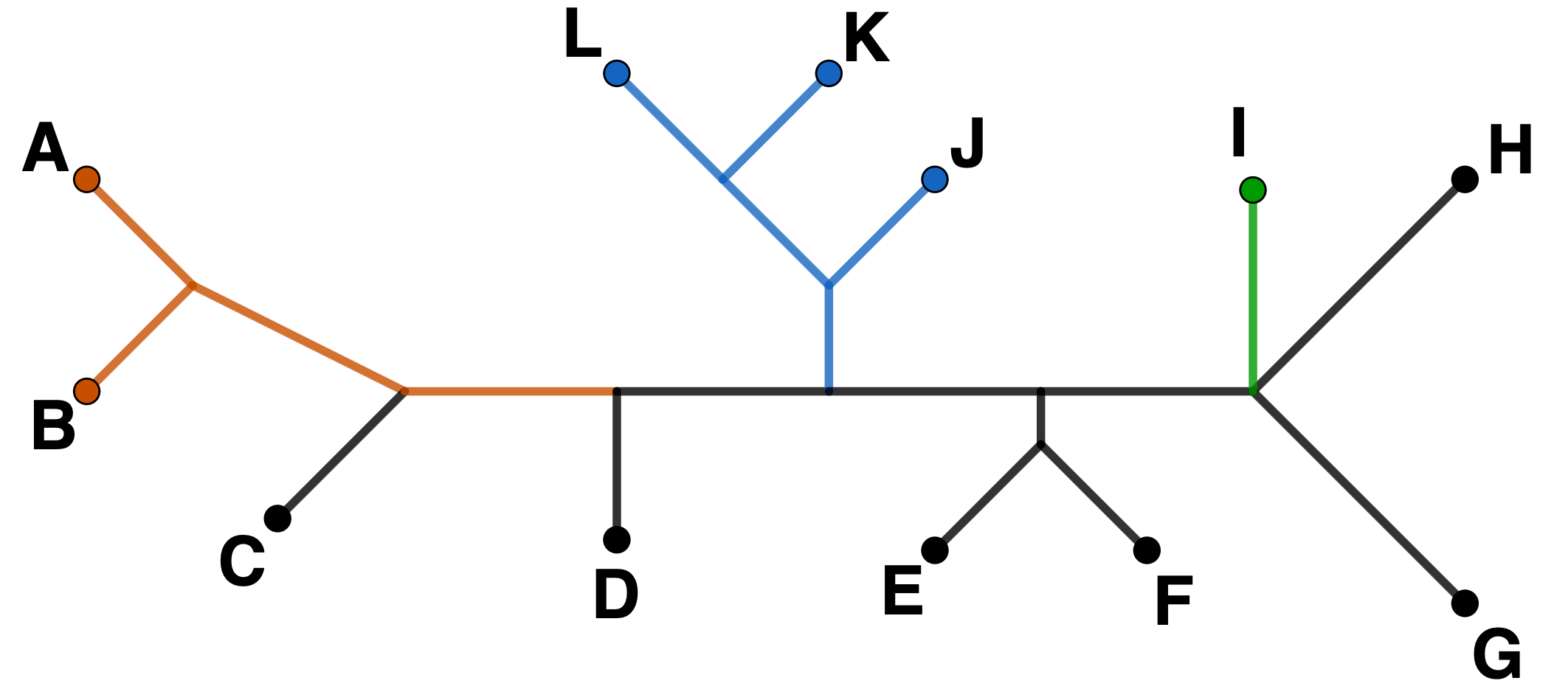}
    \caption{
    For the above tree, consider the subset of leaves $\mathcal{M} = \{A,B,I,J,K,L\}$. The independent maximal sets are $\mathcal{M}_1 = \{A,B\}$, $\mathcal{M}_2 = \{I\}$  and  $\mathcal{M}_3 = \{J, K, L\}$. $\mathcal{M}_1$ neighbours $C$, and $\mathcal{M}_2$  and  $\mathcal{M}_3$ are attached to internal edges. The set of edges belonging to each independent set is highlighted in orange, green, and blue, respectively.
    } 
    \label{fig:Example_IndependentSubTrees}
\end{figure}

This partition of the edges on $P^{\downarrow \mathcal{M}}(T)$ plays a key role in the finding the optimal tree $T^{*}$ described in Theorem \ref{Thm:First_GeodesicBest}. In the following proof this becomes apparent. 


\begin{proof}[of Theorem \ref{Thm:First_GeodesicBest}]
(i) By Corollary \ref{cor:ProjectionGood}, $$d_{\text{BHV}}^2(T_1,t) \geq d_{\text{BHV}}^2(T_1,T_1^{\perp \Me}) + d_{\text{BHV}}^2 (T_1^{\perp \Me},t)$$ for any tree $t \in \Lambda^{\Le}(T_2)$.
	Since $T_1^{\perp \mathcal{M}}, t \in \Ze_{\Me}^{\Le}$,  $d_{\text{BHV}}^2 (T_1^{\perp \mathcal{M}},t) \geq d_{\text{BHV}}^2 (T'_1,T_2)$ by Theorem \ref{thm:NoUpperShortcuts}. Therefore,
	$$d_{\text{BHV}}^2(T_1,t) \geq d_{\text{BHV}}^2(T_1,T_1^{\perp \mathcal{M}}) + d_{\text{BHV}}^2 (T'_1,T_2).$$
	
	(ii)  To construct the optimal tree $T^{*} \in \Lambda^{\mathcal{L}}(T_2)$, the positions where leaves in $\mathcal{M}$ are regrafted on $T_2$ are based on the partition of these leaves into independent maximal sets $\mathcal{M} = \mathcal{M}_1 \cup \hdots \cup \mathcal{M}_r$ in $T_1$ and the geodesic from $T'_1$ to $T_2$. The process is as follows:
	 
	   \begin{enumerate}
	   	\item For every leaf $\ell \in \mathcal{L}'$ for which there is a neighboring independent maximal set $\mathcal{M}_{i}$, regraft these at the same node in $T_2$ to which the external edge to $\ell$ is incident (see Figure \ref{fig:Mi_neighbor_leaf}).
	   	
	   	\begin{figure}
    \centering
    	\includegraphics[width = 0.45\textwidth]{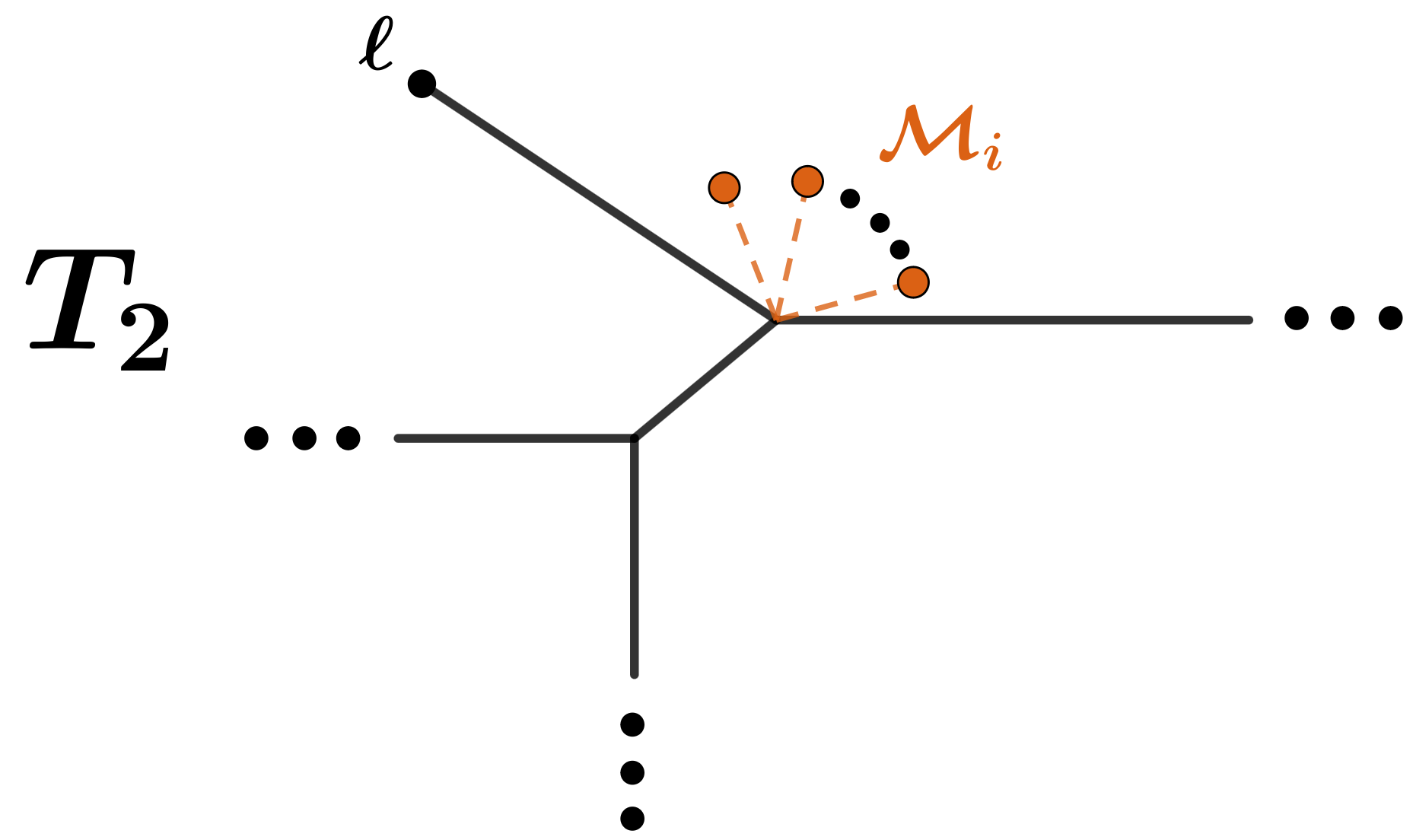}
    \caption{Regrafting an independent maximal set neighboring a leaf $\ell \in \mathcal{L}'$. $T_2 \in \mathcal{T}^{\mathcal{L}'}$ is partially shown. The new leaves $\mathcal{M}_i$ (orange) form an independent maximal set neighboring $\ell$ in another tree $T_1$, so they are regrafted to be incident to the same node as the external edge to $\ell$. The external edges of the new leaves are of length zero (dashed lines).} 
    \label{fig:Mi_neighbor_leaf}
\end{figure}
	   	
	   	\item For a common edge $p = \mathcal{L}'_1 \big| \mathcal{L}'_2 \in \mathcal{S}(T'_1) \cap \mathcal{S}(T_2)$, consider all internal edges of $T_1^{\perp \Me}$ that map to $p$ under the TDR map $\Psi_{\mathcal{L}'}$. Consider a new partition of $\mathcal{M}$ given by $\mathcal{M} = \widehat{\mathcal{M}}^p_{1} \sqcup \mathcal{M}^{p}_1 \sqcup \hdots \sqcup \mathcal{M}^{p}_{r_{p}} \sqcup \widehat{\mathcal{M}}^p_{2}$, where $\mathcal{M}^{p}_1, \hdots, \mathcal{M}^{p}_{r_{p}}$ are all the independent maximal sets that could be considered attached to $p$ in $T_1$ (including attached to the endpoint nodes of the edge) and the two sets $\widehat{\mathcal{M}}^p_{i}$ are the rest of the leaves in $\mathcal{M}$ that are on the same side of $p$ as $\mathcal{L}'_i$ in $T_1$; i.e. every edge in $T^{\perp \Me}_1$ that maps to $p$ under the TDR map is of the form 
	   	$q_{j}^{p} = \left(\mathcal{L}'_1 \cup \widehat{\mathcal{M}}^p_{1} \cup \mathcal{M}^{p}_1 \cup \hdots, \mathcal{M}^{p}_{j} \right) \big| \left(\mathcal{M}^{p}_{j+1} \cup \hdots, \mathcal{M}^{p}_{r_{p}} \cup \widehat{\mathcal{M}}^p_{2} \cup \mathcal{L}'_2\right).$
	   	When regrafting $\mathcal{M}$ onto $T_2$ to create $T^{*}$, leaves in $\widehat{\mathcal{M}}^p_{i}$ are regrafted on the same side of the edge $p$, and each $\mathcal{M}^{p}_j$ is regrafted at the same position (proportional wise) inside the edge $p$ in $T_2$. Thus, every $q_{j}^{p}$ as given above will be part of the interior edges $\mathcal{S}(T^{*})$ of $T^{*}$, with lengths given by $\left|q_{j}^{p}\right|_{T^{*}} = \frac{\left|q_{j}^{p}\right|_{T_1^{\perp \Me}}}{|p|_{T'_1}}|p|_{T_2}$ (see Figure \ref{fig:Mi_commonEdge}).
	   	\begin{figure}
    \centering
    	\includegraphics[width = 0.95\textwidth]{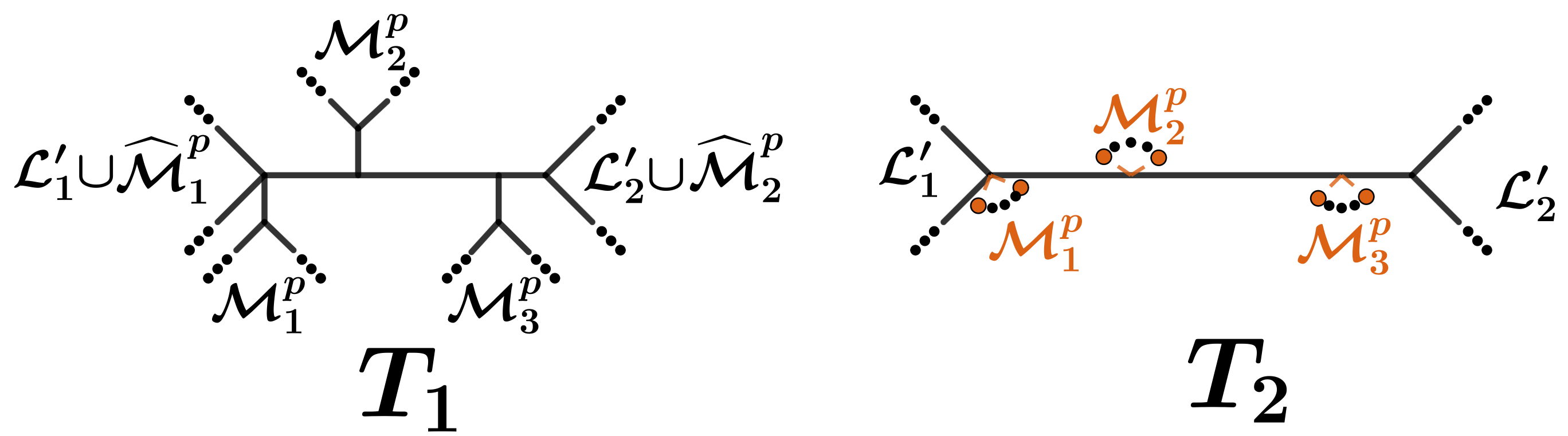}
    \caption{Regrafting independent maximal sets attached to a common edge $p \in \mathcal{S}(T'_1) \cap \mathcal{S}(T_2)$. (Left) $T_1$, containing the edges mapping to $p$ under the TDR map, featuring three independent maximal sets attached to these edges. (Right) The common edge $p$ in $T_2$, with the new leaves in the three independent maximal sets regrafted in the correct position (orange; same proportional position as their counterparts in $T_1$). Dashed lines indicate external edges of length zero.} 
    \label{fig:Mi_commonEdge}
\end{figure}
	   	
	   	\item For any independent maximal set $\mathcal{M}_{i}$ of leaves that do not fall in one of the two previous scenarios, the split $\mathcal{M}_{i} \big| \left(\mathcal{L}\setminus \mathcal{M}_{i} \right)$ is adjacent to at least one edge $q_{i} \in \mathcal{S}(T_1)$ that under $\Psi_{\mathcal{L}'}$ maps to an internal edge $p_{i} \in \mathcal{S}(T'_1)$ that is uncommon with edges in $\mathcal{S}(T_2)$. Given a support for the path space of the geodesic from $T'_1$ to $T_2$, there must be a support pair $(A,B)$  for which $p_{i} \in A$ and there exists $p'_{i} \in B$ that is the product of a nearest neighbor interchange from the node resulting from collapsing $p_{i}$ to zero. Regraft the external edges to leaves in $\mathcal{M}_{i}$ on the center of $p'_{i}$ (see Figure \ref{fig:Mi_uncommonEdge}).
	   	
	   	\begin{figure}
    \centering
    	\includegraphics[width = 0.95\textwidth]{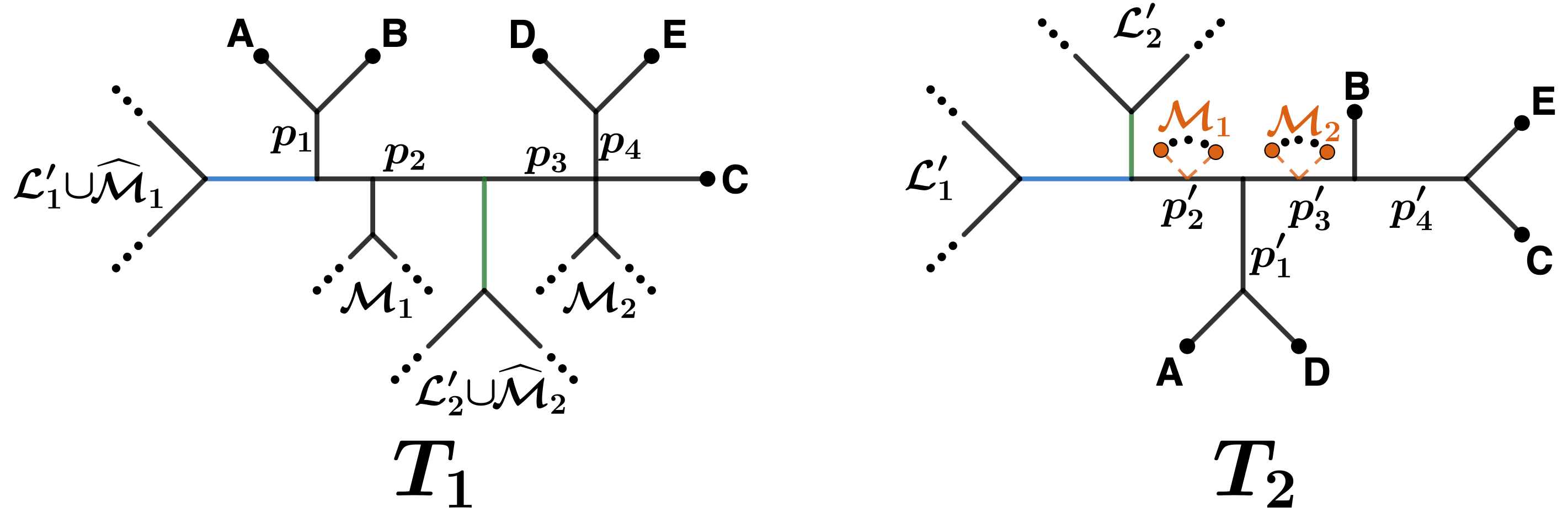}
    \caption{Regrafting independent maximal sets attached to uncommon edges for $T_1 \in \mathcal{T}^{\mathcal{L}}$ (left) and $T_2 \in \mathcal{T}^{\mathcal{L}'}$ (right). 
    After pruning $\mathcal{M}$ from $T_1^{\perp}$, the edges $\left\{A, B, C, D, E\right\} \cup \mathcal{L}'_2 \big| \mathcal{L}'_1$ (blue) and $\left\{A, B, C, D, E\right\} \cup \mathcal{L}'_1 \big| \mathcal{L}'_2$ (green) are common, while the edges $p_1, p_2, p_3, p_4 \in \mathcal{S}(T'_1)$ and $p'_1, p'_2, p'_3, p'_4 \in \mathcal{S}(T_2)$ are uncommon. In $T_1$, the independent maximal set $\mathcal{M}_1$ is adjacent to two edges mapping to $p_2$ under the TDR map, and $\mathcal{M}_2$ is adjacent to edges that map to $p_3$ and $p_4$. Assuming the support for the geodesic from $T'_1$ to $T_2$ includes the support pairs $\left(\{p_2\},\{p'_2\}\right)$, $\left(\{p_4\}, \{p'_4\}\right)$, and $\left(\{p_1, p_3\}, \{p'_1, p'_3\}\right)$, $p'_2$ can be selected to regraft $\mathcal{M}_1$ and $p'_3$ to regraft $\mathcal{M}_2$. These leaves are represented in orange in the correct position for the regraft, with dashed lines indicating the external edges of length zero.} 
    \label{fig:Mi_uncommonEdge}
\end{figure}
	   \end{enumerate}

       The tree $T^{*}$ created through these regraftings contains all leaves in $\mathcal{M}$ and falls in $\Lambda^{\mathcal{L}}(T_2)$. Moreover, the BHV distance from $T_1$ to $T^{*}$ will precisely be $\sqrt{d^2_{\text{BHV}}(T_1, T_1^{\perp \mathcal{M}}) + d^2_{\text{BHV}}(T'_1, T_2)}$. To see this, consider again the edges $P^{\downarrow \mathcal{M}}(T_1)$. If an edge $s_{\mathcal{M}'} \in P^{\downarrow \mathcal{M}}(T_1)$ takes the form $\mathcal{M}'\big| \mathcal{L}\setminus \mathcal{M}'$, by definition $\mathcal{M}' \subseteq \mathcal{M}_i$ for some $i = 1, \hdots r,$ and since all leaves in the maximal independent sets are regrafted together when constructing $T^{*}$, then $s_{\mathcal{M}'}$ is pairwise compatible with $\mathcal{S}(T^{*})$. Similarly, for an edge $s^{\ell}_{\mathcal{M}'} \in P^{\downarrow \mathcal{M}}(T_1)$ of the form $\mathcal{M}' \cup \{\ell\} \big| \mathcal{L}\setminus \left(\mathcal{M} \cup \{\ell\}\right)$, we have $\mathcal{M}'\subseteq \mathcal{M}_i$ for the independent maximal set $\mathcal{M}_i$ neighboring $\ell$. In $T^{*}$ all external edges to leaves in $\mathcal{M}_i$ are adjacent to the external edge to $\ell$ (as in Figure \ref{fig:Mi_neighbor_leaf}). This implies $s^{\ell}_{\mathcal{M}'}$ is compatible with every edge in $T^{*}$. Since $P^{\downarrow \mathcal{M}}(T_1)$ is pairwise compatible with $\mathcal{S}(T^{*})$, the lengths all edges $s_{\mathcal{M}'}$ and $s^{\ell}_{\mathcal{M}'}$ will gradually diminish to zero along the geodesic. Thus, the contribution of these edges to the geodesic length is $||P^{\downarrow \mathcal{M}}(T_1)|| = d_{\text{BHV}}(T_1, T_1^{\perp \mathcal{M}})$.

    We now focus on the remaining internal edges ($\mathcal{S}(T_1) \cup \mathcal{S}(T^{*}) \setminus P^{\downarrow \mathcal{M}}(T_1)$) and their respective contribution to the length of the geodesic. All internal edges in $\mathcal{S}(T_1)$ that are not in $P^{\downarrow \mathcal{M}}(T_1)$ belong to $\mathcal{S}(T_1^{\perp})$, and have the same lengths. Step 2 of the construction of $T^{*}$ guarantees $q \in \mathcal{S}(T_1^{\perp}) \cap \mathcal{S}(T^{*})$ if and only if $\Psi_{\mathcal{L}'}(q) \in \mathcal{S}(T'_1) \cap \mathcal{S}(T_2)$. The contribution to the length of the geodesic of all common edges $q^{p}_1,...,q^{p}_{r}$ in $\mathcal{S}(T_1^{\perp}) \cap \mathcal{S}(T^{*})$  that map into a common edge $p$ is given by
	   \begin{equation*}
	   		\begin{split}
	   		\sum_{k = 1}^{r} (|q^{p}_k|_{T_1^{\perp}} - |q^{p}_k|_{T^{*}})^2 &= \sum_{k = 1}^{r} \left(|q^{p}_k|_{T_1^{\perp}} - \frac{\left|q_{k}^{p}\right|_{T_1^{\perp}}}{|p|_{T'_1}}|p|_{T_2}\right)^2\\
	   		 &= \left(1 - \frac{|p|_{T_2}}{|p|_{T_1^{\perp}}}\right)^2 \sum_{k=1}^{r} |q^{p}_{k}|^2_{T_1^{\perp}} \notag \\
			 &= \left(|p|_{T'_1} - |p|_{T_2}\right)^2.
	   		\end{split}
		\end{equation*}	 
		The contribution of these edges to the length of the geodesic from $T_1$ to $T^{*}$ is the same as the contribution of $p$ in the length of the geodesic from $T'_1$ to $T_2$. 
		
        Finally, consider the contribution of uncommon edges between $T_1$ and $T^{*}$. 
        Denote the support of the geodesic from $T'_1$ to $T_2$ by $\mathcal{A} = \{A_1, \hdots, A_k\}$ and $\mathcal{B} = \{B_1, \hdots, B_k\}$. 
        For each $A_i$, consider the set of edges in $T_1^{\perp}$ that map to an edge in $A_i$, $\left. \Psi_{\mathcal{L}'}^{-1}(A_i) \right|_{\mathcal{S}(T_1^{\perp})}$, and similarly, the set of edges in $T^{*}$ that map to edges in $B_i$, $\left. \Psi_{\mathcal{L}'}^{-1}(B_i) \right|_{\mathcal{S}(T^{*})}$. 
        The positions for the independent maximal sets in Step 3 are selected to ensure that $\left(\left. \Psi_{\mathcal{L}'}^{-1}(A_1) \right|_{\mathcal{S}(T_1^{\perp})}, \ldots, \left. \Psi_{\mathcal{L}'}^{-1}(A_k) \right|_{\mathcal{S}(T_1^{\perp})} \right)$ and $\left(\left. \Psi_{\mathcal{L}'}^{-1}(B_1) \right|_{\mathcal{S}(T^{*})}, \ldots, \left. \Psi_{\mathcal{L}'}^{-1}(B_k) \right|_{\mathcal{S}(T^{*})}\right)$ satisfy the pairwise compatible condition for path space supports.
        Since $||A_i||_{T'_1} = ||\left. \Psi_{\mathcal{L}'}^{-1}(A_i) \right|_{\mathcal{S}(T_1^{\perp})}||_{T_1^{\perp}}$ and $||B_i||_{T_2} = ||\left. \Psi_{\mathcal{L}'}^{-1}(B_i) \right|_{\mathcal{S}(T^{*})}||_{T^{*}}$ the condition (P2) is also satisfied. Consequently, the contribution of the support pair $\left(\left. \Psi_{\mathcal{L}'}^{-1}(A_i) \right|_{\mathcal{S}(T_1^{\perp})}, \left. \Psi_{\mathcal{L}'}^{-1}(B_i) \right|_{\mathcal{S}(T^{*})}\right)$ to the length of the geodesic from $T_1$ to $T^{*}$ is the same as the contribution of the edges $(A_i, B_i)$ to the length of the geodesic from $T'_1$ to $T_2$. Thus, the internal edges in $\mathcal{S}(T_1) \cup \mathcal{S}(T^{*}) \setminus P^{\downarrow \mathcal{M}}(T_1)$ contribute to the length of the geodesic from $T_1$ to $T^{*}$ an amount exactly equal to that of the internal edges in $\mathcal{S}(T'_1) \cup \mathcal{S}(T_2)$'s contribution to the geodesic from $T'_1$ to $T_2$. The result follows.

\end{proof}

\begin{proof}[of Lemma \ref{lemma:DirectPruning}]
Consider a general path from $T_1$ to $T_2$ where the pruning of $\mathcal{M}$ occurs in the sprouting space $\Lambda^{\mathcal{L}}(X)$ of a tree $X \in \mathcal{T}^{\mathcal{L}'}$. By Theorem \ref{Thm:First_GeodesicBest}, the shortest the section of the path from $T_1$ to $X$ is $\sqrt{d^2_{\text{BHV}}(T_1,T_1^{\perp \mathcal{M}}) + d^2_{\text{BHV}}(T_1',X)}$, where $T_1'= \psi(T_1^{\perp \mathcal{M}},\mathcal{M})$. The shortest path from $X$ to $T_2$ completely contained in $\mathcal{T}^{\mathcal{L}'}$ is simply the geodesic with length $d_{\text{BHV}}(X,T_2)$. On the other hand, again by Theorem \ref{Thm:First_GeodesicBest}, there is a path from $T_1$ to $T_2$ where the pruning of $\mathcal{M}$ is at a tree in $\Lambda^{\mathcal{L}}(T_2)$. This path has length $\sqrt{d^2_{\text{BHV}}(T_1,T_1^{\perp \mathcal{M}}) + d^2_{\text{BHV}}(T_1',T_2)}$. 

	By the triangle inequality, $d_{\text{BHV}}(T'_1,X) + d_{\text{BHV}}(X,T_2) \geq d_{\text{BHV}}(T'_1, T_2)$, and thus 
	$$d^2_{\text{BHV}}(T'_1,X) + 2 d_{\text{BHV}}(T'_1,X)d_{\text{BHV}}(X,T_2) + d^2_{\text{BHV}}(X,T_2) \geq d^2_{\text{BHV}}(T'_1, T_2).$$ 
	It is clear that 
	$$\sqrt{d^2_{\text{BHV}}(T_1,T_1^{\perp \mathcal{M}}) + d^2_{\text{BHV}}(T_1',X)} \geq d_{\text{BHV}}(T'_1,X),$$ 
	and so
	$$d^2_{\text{BHV}}(T'_1,X) + 2 d_{\text{BHV}}(X,T_2)\sqrt{d^2_{\text{BHV}}(T_1,T_1^{\perp \mathcal{M}}) + d^2_{\text{BHV}}(T_1',X)} + d^2_{\text{BHV}}(X,T_2) \geq d^2_{\text{BHV}}(T'_1, T_2).$$
	Adding $d^2_{\text{BHV}}(T_1,T_1^{\perp \mathcal{M}})$ to both sides, we have that 
	$$\sqrt{d^2_{\text{BHV}}(T_1,T_1^{\perp \mathcal{M}}) + d^2_{\text{BHV}}(T_1',X)} + d_{\text{BHV}}(X,T_2) \geq \sqrt{d^2_{\text{BHV}}(T_1,T_1^{\perp \mathcal{M}}) + d^2_{\text{BHV}}(T_1',T_2)}.$$
\end{proof}

\begin{proof}[of Lemma \ref{lemma:OnePruningRatherThanTwo}]
	Consider any path from $T_1$ to $T_2$ where a pruning of $\mathcal{M}_1$ is performed in the sprouting space $\Lambda^{\mathcal{L}}(X)$ for some tree $X \in \mathcal{T}^{\mathcal{L}_1}$, where $\mathcal{L}_1 = \mathcal{L}\setminus \mathcal{M}_1$ and then the pruning of $\mathcal{M}_2$ is performed in some other tree $Y \in \mathcal{T}^{\mathcal{L}_1}$. Assume $Y \in \Lambda^{\mathcal{L}_1}(T_2)$, so that the path from $X$ to $T_2$ is as short as possible (Lemma \ref{lemma:DirectPruning}). Lemma \ref{lemma:DirectPruning} also implies that going from $T_1$ to $Y$ by performing the pruning of $\mathcal{M}_1$ at a tree $Y^{\uparrow} \in \Lambda^{\mathcal{L}}(Y)$ produces a shorter path. 
    This implies $\mathcal{M}$ is prunable from $Y^{\uparrow}$ and we can conclude $\psi(Y^{\uparrow},\mathcal{M}) = \psi(Y, \mathcal{M}_2) = T_2$. 
\end{proof}

\begin{proof}[of Theorem \ref{Thm:shorterPathsBelow}] Any path of this form will pass through a tree $X \in \mathcal{T}^{\mathcal{L}'}$. Lemmas \ref{lemma:DirectPruning} and  \ref{lemma:OnePruningRatherThanTwo} imply that the shortest paths from $T_i$ to $X$ (for $i= 1, 2$) that exclusively perform leaf prunings are given by performing the pruning of $\mathcal{M}_i$ at trees in the sprouting space  $\Lambda^{\mathcal{L}_i}(X)$. The length of these paths are $\sqrt{d^2_{\text{BHV}}(T_i, T_i^{\perp \mathcal{M}_i}) + d^2_{\text{BHV}}(T_i',X)}$. Thus, the shortest paths from $T_1$ to $T_2$, going strictly to lower levels, passing through $X$ and then going strictly through higher levels to $T_2$ are of length
	\begin{equation*}
		\sqrt{d^2_{\text{BHV}}(T_1, T_1^{\perp \mathcal{M}_1}) + d^2_{\text{BHV}}(T_1',X)} + \sqrt{d^2_{\text{BHV}}(X,T_2') + d^2_{\text{BHV}}(T_2^{\perp \mathcal{M}_2},T_2)}. 
	\end{equation*}
The values $d^2_{\text{BHV}}(T_1, T_1^{\perp \mathcal{M}_1}) $ and $d^2_{\text{BHV}}(T_2^{\perp \mathcal{M}_2},T_2)$ are independent from the choice of $X \in \mathcal{T}^{\mathcal{L}'}$. For simplicity, denote $\rho_i = d_{\text{BHV}}(T_i, T_i^{\perp \mathcal{M}_i})$. Since $T_1', X$ and $T_2'$ are all in the same BHV space, $d_{\text{BHV}}(T_1',X) + d_{\text{BHV}}(X,T_2') \geq d_{\text{BHV}}(T_1',T_2')$. Let $y_1 = d_{\text{BHV}}(T_1',X)$, $y_2 = d_{\text{BHV}}(X,T_2')$ and $d = d_{\text{BHV}}(T_1',T_2')$. The problem of finding $X$ that produces the shortest path between $T_1$ and $T_2$ is now equivalent to finding values $y_1$ and $y_2$ that minimize $\sqrt{\rho_1^2 + y_1^2} + \sqrt{\rho_2^2 + y_2^2}$ subject to $y_1 + y_2 \geq d$. Consider a geometrical approach by considering points $\mathbf{a} = (z_1,\rho_1)$ and $\mathbf{b} = (z_2,-\rho_2)$ in $\mathbb{R}^2$, where $|z_1-z_2| = d$. Consider the piecewise linear path from $\mathbf{a}$ to $\mathbf{b}$ crossing the x-axis at a point $\mathbf{x} = (x,0)$. The length of this path is given by $f(x) = \sqrt{\rho_1^2 + (z_1 - x)^2} + \sqrt{\rho_2^2 + (x-z_2)^2}$. But the shortest of these paths is the direct line from $\mathbf{a}$ to $\mathbf{b}$, that crosses the x-axis at $\left( \frac{\rho_1 z_2 + \rho_2 z_1}{\rho_1 + \rho_2},0\right)$ (see Figure \ref{fig:LittleGeometry}). This implies the minimum of $f(x)$ is reached by $x^* = \frac{\rho_1 z_2 + \rho_2 z_1}{\rho_1 + \rho_2}$, where $|z_i - x^{*}| = \frac{\rho_i}{\rho_1 + \rho_2} |z_1 - z_2|$ for $i=1, 2$, and the minimum value is $f(x^*) = \sqrt{(\rho_1 + \rho_2)^2 + d^2}$. 
Equivalently, the value $\sqrt{\rho_1^2 + y_1^2} + \sqrt{\rho_2^2 + y_2^2}$ is lower-bounded by $\sqrt{(\rho_1 + \rho_2)^2 + d^2}$ and it is achieved when $y_1 = \frac{\rho_1}{\rho_1 + \rho_2} d$ and $y_2 = \frac{\rho_2}{\rho_1 + \rho_2} d$. This, in turn, is reached on the tree $X^{*}$ on the geodesic from $T'_1$ to $T'_2$, at a distance $\frac{\rho_1}{\rho_1 + \rho_2} d_{\text{BHV}}(T'_1, T'_2)$ from $T'_1$. 

\begin{figure}
    \centering
    	\includegraphics[width = 0.8\textwidth]{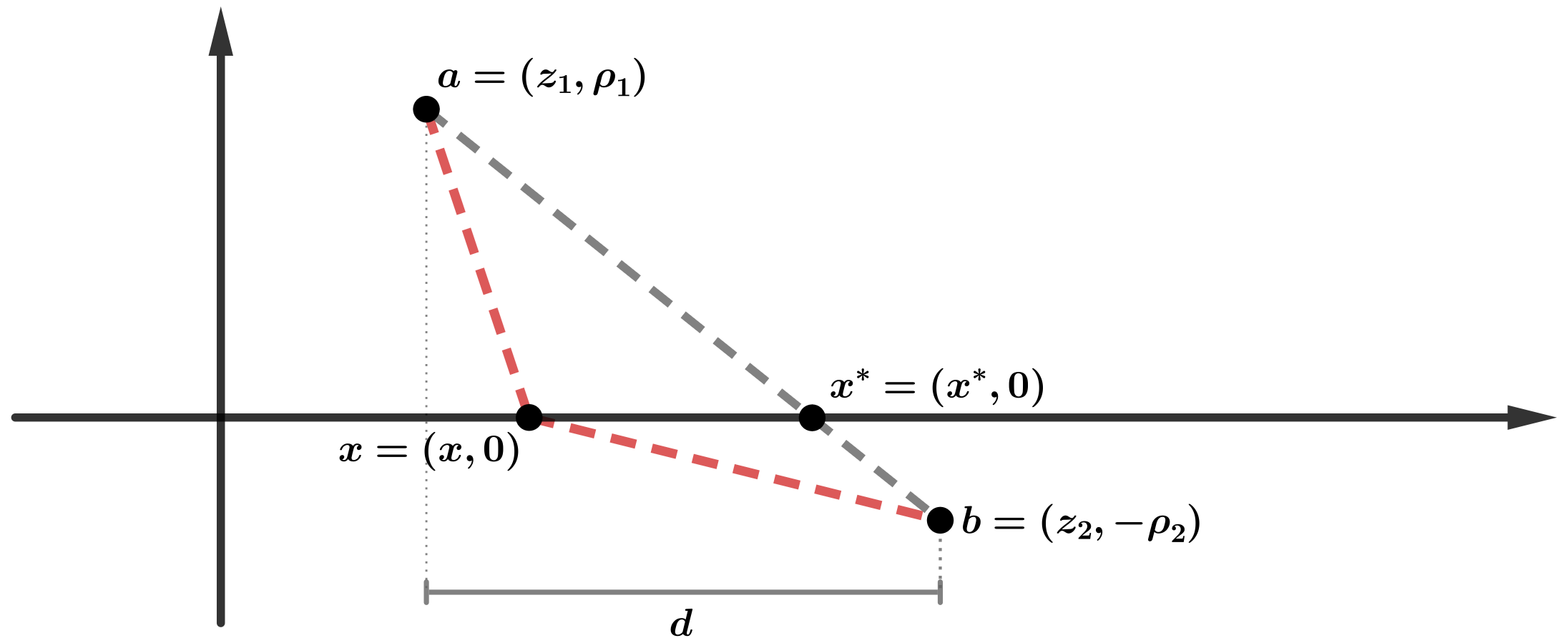}
    \caption{A two-dimensional representation of the minimizing function $f(x) = \sqrt{\rho_1^2 + (z_1 - x)^2} + \sqrt{\rho_2^2 + (z_2 - x)^2}$. The length of the path from $\mathbf{a} = (z_1,\rho_1)$ to $\mathbf{b} = (z_2, \rho_2)$ passing through $(0,x)$ (dashed red) coincides with $f(x)$, and any path is longer than the direct segment (dashed gray).}
    \label{fig:LittleGeometry}
\end{figure}
\end{proof}

\begin{proof}[of Lemma \ref{lemma:LiftingGeodesics}]
	This lemma follows from Theorem \ref{Thm:First_GeodesicBest}. To see this, take $T_1 \in \Lambda^{\Le}(T'_1)$ to be any regraft of $\Me = \Le \setminus \Le'$ onto $T'_1$. By Theorem \ref{Thm:First_GeodesicBest}, there exists a tree $T_2 \in \Lambda^{\Le}(T'_2)$ such that 
	$$d_{\text{BHV}}(T_1, T_2) = \sqrt{||P^{\downarrow \Me}(T_1)||^2 + d_{\text{BHV}}^2 (T'_1,T'_2)},$$
	but since $T_1 \in \Ze^{\Le}_{\Me}$ and thus $||P^{\downarrow \Me}(T_1)||^2 = 0$, and so $d_{\text{BHV}}(T_1, T_2) = d_{\text{BHV}}(T'_1, T'_2).$
\end{proof}

\begin{proof}[of Theorem \ref{Thm:4}]
	Any path constructed using Theorem \ref{Thm:shorterPathsBelow} and Lemmas \ref{lemma:DirectPruning} and \ref{lemma:OnePruningRatherThanTwo} consists of segments fully contained in some BHV level going from tree $t^{\uparrow}_{i}$ to $t^{\downarrow}_{i+1}$. By Lemma \ref{lemma:LiftingGeodesics}, each of these segments can be lifted to an equal length segment in $\Te^{\Ne}$. 
\end{proof}

\section{Supplementary Material: Proof that the Towering distance is a metric}

We prove that \textit{non-deformed trees} are correctly differentiated by the distance function in the Towering space. That is, if $T_1$ and $T_2$ are distinct trees with all edges of strictly positive length, then $d(T_1,T_2) > 0$.

\begin{lemma}
    \label{lemma:LeafOperationOne}
    Consider a tree $T \in \Ze_{\Me}^{\Le}$ and take $T' = \psi(T,\Me)$. Taking a proper subset $\Ke \subset \Le'$ of the leaves $\Le'= \Le \setminus \Me$ of $T'$, then $||P^{\downarrow \Ke}(T')||^2 = ||P^{\downarrow \Ke \cup \Me}(T)||^2$ and $\psi(T^{\perp \Ke \cup \Me},\Ke \cup \Me) = \psi(T^{'\perp \Ke}, \Ke)$.
\end{lemma}

\begin{proof}
    By the definition of $T'$, we know that $\Se(T') = \Psi_{\Le'}(\Se(T)) = \{\Psi_{\Le'}(s) |  s\in \Se(T)\}$. We start by showing that there is a natural partition between the edges that prevent $\Ke$ from being prunable from $T$ and $T'$, and the edges that do not prevent these leaves from being prunable. That is, we prove $\Psi_{\Le'} (P^{\downarrow \Ke \cup \Me}(T)) = P^{\downarrow \Ke}(T')$ and $\Psi^{-1}_{\Le'} (P^{\downarrow \Ke}(T')) \cap \Pe(T) = P^{\downarrow \Ke \cup \Me}(T)$. 

    Consider a split $s = \mathcal{G} \big| \Le \setminus \mathcal{G} \in P^{\downarrow \Ke\cup \Me} (T)$. There are two cases: 
    \begin{enumerate}
        \item $s = \{\ell\} \big| \Le \setminus \{\ell\}$ for $\ell \in \Ke \cup \Me$ and with positive length. Since $T \in \Ze^{\Le}_{\Me}$, all external edges to leaves in $\Me$ are of size zero. This implies $\ell \in \Ke$. Thus, $\Psi_{\Le'}(s) = \{\ell\} \big| \Le' \setminus \{\ell\}$ and it belongs to $P^{\downarrow \Ke}(T')$. 
        \item $s$ is an internal edge and we can assume w.l.o.g. $|\mathcal{G}\setminus \left(\Ke \cup \Me\right)| < 2$. But since $\Se(T') = \Psi_{\Le'}(\Se(T))$, then we know $s' = \Psi_{\Le'}(s) = \mathcal{G}\setminus \Me \big| \Le' \setminus \left(\mathcal{G}\setminus \Me\right) \in \Se(T')$. We note that $\left(\mathcal{G}\setminus \Me\right) \setminus \Ke = \mathcal{G}\setminus \left(\Ke \cup \Me\right)$, and so conclude $s' \in P^{\downarrow \Ke}(T')$. 
    \end{enumerate}

    Similarly, take $s = \mathcal{G} \big| \Le \setminus \mathcal{G}$ such that $s'= \Psi_{\Le'}(s) = \mathcal{G}\setminus \Me \big| \Le' \setminus \left(\mathcal{G}\setminus \Me\right) \in P^{\downarrow \Ke}(T')$. Assume w.l.o.g that $|\left(\mathcal{G}\setminus \Me\right)\setminus \Ke| < 2$, which implies $s \in P^{\downarrow \Ke \cup \Me}(T)$. 

    Since we have this correspondence between the sets $P^{\downarrow \Ke \cup \Me}(T)$ and $P^{\downarrow \Ke}(T')$ through the TDR map, we conclude $||P^{\downarrow \Ke}(T')||^2 = ||P^{\downarrow \Ke \cup \Me}(T)||^2$. 

    Moreover, since $\Se(T^{\perp \Ke \cup \Me}) = \left\{s \in \Se(T) \mid s \notin P^{\downarrow \Ke \cup \Me}(T) \right\}$ and $\Se(T^{'\perp \Ke}) = \left\{s \in \Se(T') \mid s \notin P^{\downarrow \Ke}(T') \right\}$, all splits in $T$ mapping to edges in $T^{'\perp \Ke}$ under the TDR map $\Psi_{\Le'}$ are still preserved when projecting onto $\Ze^{\Le}_{\Ke \cup \Me}$ to give the tree $T^{\perp \Ke \cup \Me}$ with identical edge lengths to $T'$. This implies $\psi(T^{\perp \Ke \cup \Me}, \Me) = T^{'\perp \Ke}$, and in turn this implies $\psi(T^{\perp \Ke \cup \Me},\Ke \cup \Me) = \psi(T^{'\perp \Ke}, \Ke)$. 
\end{proof}

\begin{lemma}
    \label{lemma:EquivalenceSprouting}
    Consider a tree $T \in \Te^{\Le}$ with all external edges strictly positive. If $T' \in \mathcal{L'}$ satisfies $T' \simeq T$, then $\Le \subset \Le'$, $\mathcal{K} = \Le'\setminus \Le$ is prunable from $T'$, and $\psi(T', \Ke) = T$.
\end{lemma}

\begin{proof}
    Consider a sequence of trees connecting $T$ and $T'$, $T  \simeq t_1 \simeq \hdots \simeq t_k \simeq T'$ where each tree in the sequence is the result of a leaf operation in the previous tree, meaning $t_{i} = \psi(t_{i-1}, \Me_{i})$ or $t_{i-1} = \psi(t_{i}, \Me_i)$ for some set of leaves $\Me_i$ for all $i$. Note that leaf operations cannot modify the length of non-zero external edges, which implies that for each tree in the sequence the external edges to the leaves in $\Le$ are of positive length (they are same length as in $T$) and are not prunable. Thus, $\Le$ is part of the leaves for each tree in the sequence, including $T'$. We proceed by induction on $k$ to prove all other leaves are prunable and that their leaf prune returns $T$.

    For $k = 1$, we have either $T' = \psi(T, \Me_1)$ or $T = \psi(T', \Me_1)$. Since all external edges in $T$ are strictly positive, none of the leaves in $\Le$ are prunable from $T$. Thus $T' = \psi(T, \Me_1)$ is not possible and we conclude $\Me_1 = \Le'\setminus \Le$ is prunable from $T'$ and $T = \psi(T', \Me_1)$. 

    Now, assume the lemma is true for any tree connected through a sequence of size $k$. Thus, for $t_{k} \in \Te^{\Le_{k}}$,  $\Le \subset \Le_{k}$, $\Ke_{k} = \Le_{k}\setminus \Le$ is prunable from $t_{k}$ and $T = \psi(t_{k},\Ke_{k})$. There are two options: 

    \begin{enumerate}
        \item If $T' = \psi(t_{k}, \Me_k)$, since none of the leaves in $\Le$ are prunable from $t_{k}$, we know $\Me_{k} \subset \Le_{k}\setminus\Le$ and this implies $\Ke_{k} = \Ke \cup \Me_{k}$. Moreover, since $\Ke_{k}$ is prunable from $t_{k}$, $\Ke = \Ke_{k} \setminus \Me_{k}$ is prunable from $T'$, and by Lemma \ref{lemma:LeafOperationOne} we may conclude $\psi(T',\Ke) = \psi(t_{k},\Ke_{k}) = T$;

        \item If $t_{k} = \psi(T', \Me_k)$ then $\Ke = \Ke_{k} \cup \Me_{k}$, which implies $\Ke$ is prunable from $T'$ and $\psi(T',\Ke) = \psi(\psi(T',\Me_k),\Ke_{k}) = \psi(t_{k},\Ke_k) = T$.
    \end{enumerate}
\end{proof}

In other words, $T$ can serve as a ``representative'' of its equivalence class in the towering space, and we can say this equivalence class $[T]$ is the union of all sprouting spaces of $T$ in BHV level above $\Te^{\Le}$; i.e. $[T] = \bigcup_{\mathcal{L} \subset \Le' \subseteq \mathcal{N}} \Lambda^{\Le'}(T)$. For the purpose of the following results, we write $d_{\text{BHV}}(T_1, [T_2]) = \inf_{t \simeq T_2} d_{\text{BHV}}(T_1, t)$. Note this value will be infinite if there is no element of $[T_2]$ in the same BHV space as $T_1$. We use the notation $d_{\text{BHV}}([T_1],[T_2]) = \inf_{t_1 \simeq T_1, t_2 \simeq T_2} d_{BHV}(t_1,t_2)$.

\begin{lemma}
    \label{lemma:UnionOnEquivalence}
    Fix $T \in \Te^{\Le}$ with strictly positive external edge lengths and any trees satisfying $X_1 \simeq X_2$. Then there exists $\varepsilon >0$ such that $d_{\text{BHV}}([X_1],[T]) < \varepsilon$ implies $d_{\text{BHV}}(X_1, [T]) = d_{\text{BHV}}(X_2, [T])$.
\end{lemma}

\begin{proof}
    Let $\varepsilon = \frac{1}{2}\min \{|s|_{T} : s \text{ is a external edge of $T$}\}$. Since the lengths of non-zero external edges in trees cannot be modified by leaf operations, any tree $T' \in \Lambda^{\Le'}(T)$ has the same external edge lengths as $T$. For any such tree, $d_{\text{BHV}}(X',T') < \varepsilon$ would imply that $X'$ has strictly positive lengths for external edges of leaves in $\Le$ as well. This also implies that any tree $X \simeq X'$ would also have these leaves, and with strictly positive lengths. 

    Consider $X_1 \in \Te^{\Le_1}$ and $X_2 \in \Te^{\Le_2}$ with $\Le \subseteq \Le_1, \Le_2$. From Lemma \ref{lemma:EquivalenceSprouting} we have $d_{\text{BHV}}(X_i,[T]) = d_{\text{BHV}}(X_i,\Lambda^{\Le_i}(T))$, which by Theorem \ref{Thm:First_GeodesicBest} is 
    $$d_{\text{BHV}}(X_i,[T]) = \sqrt{||P^{\downarrow \Me_i}(X_i)||^2 + d_{\text{BHV}}^2 (X'_i,T)},$$
    where $\Me_i = \Le_i \setminus \Le$ and $X'_i = \psi(X_i^{\perp \Me_i},\Me_i)$. Since $X_1 \simeq X_2$, there is a sequence of trees $X_1 = t_0 \simeq t_1 \simeq \hdots \simeq t_{k-1} \simeq t_{k} = X_2$, where each pair $t_{i-1}$ and $t_i$ is connected by a direct leaf operation such that either $t_{i-1} = \psi(t_{i},\Ke_{i})$ or $t_{i} = \psi(t_{i-1},\Ke_{i})$ for some set of leaves $\Ke_{i}$.  

    Since all trees in the sequence are part of the equivalence class of $X_1$ and $X_2$, all have $\Le$ in their leaf sets, and these leaves have strictly positive external lengths. For each tree in the sequence, say $t_i \in \mathcal{T}^{\Le'_i}$ and take $\Me'_i = \Le'_i \setminus \Le$. For any pair $t_{i-1}$ and $t_i$ assume w.l.o.g $t_{i-1} = \psi(t_{i},\Ke_{i})$. In this case $\Me'_{i} = \Ke_i \cup \Me'_{i-1}$ and so by Lemma \ref{lemma:LeafOperationOne} we have $||P^{\downarrow \Me'_{i-1}}(t_{i-1})||^2 = ||P^{\downarrow \Me'_{i}}(t_i)||^2$ and $\psi(t^{\perp \Me'_{i-1}}_{i-1}, \Me'_{i-1}) = \psi(t^{\perp \Me'_{i}}_{i}, \Me'_{i})$. 

    Since this is true for each pair along the sequence between $X_1$ and $X_2$, we have that $||P^{\downarrow \Me_1}(X_1)||^2 = ||P^{\downarrow \Me_2}(X_2)||^2$ and $X'_1 = X'_2$, which implies $d_{\text{BHV}}(X_1, [T]) = d_{\text{BHV}}(X_2, [T])$.
\end{proof}

\begin{theorem}
Given a tree $T_1$ with all external edges of strictly positive length and $T_2$ such that $d(T_1,T_2) = 0$, $T_1 \simeq T_2$.  
\end{theorem}

\begin{proof}
    Let $\varepsilon_1 > 0$ be as in Lemma \ref{lemma:UnionOnEquivalence} for $T_1$. Since $d(T_1,T_2) = 0$, we can find a finite sequence $t_1, t'_1, t_2, t'_2, \hdots, t_k, t'_k$ with $T_1 \simeq t_1, t'_i \simeq t_{i+1} \forall i=1,\hdots, k-1, t'_k \simeq T_2$ such that 

        $$\sum_{i=1}^{k} d_{\text{BHV}}(t_i,t'_i) < \varepsilon_1.$$
   We upper-bound $d_{\text{BHV}}([T_1],[T_2])$ by $\sum_{i=1}^{k} d_{\text{BHV}}(t_i,t'_i)$ by proving $d_{\text{BHV}}(t'_k, [T_1]) \leq \sum_{i=1}^{k} d_{\text{BHV}}(t_i,t'_i)$. We proceed by induction. 

    Since $t_1 \in [T_1]$, we readily have $d_{\text{BHV}}(t'_1, [T_1]) \leq d_{\text{BHV}}(t_1,t'_1)$. Now assume that for some value $j < k$ it is true $d_{\text{BHV}}(t'_j, [T_1]) \leq \sum_{i=1}^{j} d_{\text{BHV}}(t_i,t'_i) < \varepsilon_1$. By Lemma \ref{lemma:UnionOnEquivalence} we have $d_{\text{BHV}}(t_{j+1},[T_1]) = d_{\text{BHV}}(t'_{j},[T_1])$, and so 

    \begin{equation*}
        \begin{split}
            d_{\text{BHV}}(t'_{j+1}, [T_1]) &\leq d_{\text{BHV}}(t_{j+1}, t'_{j+1}) + d_{\text{BHV}}(t_{j+1},[T_1]) \\
            &= d_{\text{BHV}}(t_{j+1}, t'_{j+1}) + d_{\text{BHV}}(t'_{j},[T_1]) \\ &\leq \sum_{i=1}^{j+1} d_{\text{BHV}}(t_i,t'_i).
        \end{split}
    \end{equation*}
    From this we have 
    $$d_{\text{BHV}}([T_1],[T_2]) \leq \sum_{i=1}^{k} d_{\text{BHV}}(t_i,t'_i).$$
    Since the sum can be made arbitrarily small, we have $d_{\text{BHV}}([T_1],[T_2]) = 0$. Since $[T_1] = \bigcup_{\Le} \Lambda^{\Le}(T_1)$, which is a closed set, this implies $T_2 \simeq T_1$. 
    
\end{proof}

\bibliographystyle{biometrika}

\bibliography{bibliography.bib}
\end{document}